\pdfoutput=1
\RequirePackage{ifpdf}
\ifpdf 
\documentclass[pdftex]{sigma}
\else
\documentclass{sigma}
\fi

\newtheorem{thm}{Theorem}[section]

\newtheorem{lem}[thm]{Lemma}
\newtheorem{prop}[thm]{Proposition}

\theoremstyle{definition}
\newtheorem{Def}[thm]{Definition}
\newtheorem{rem}[thm]{Remark}
\newtheorem{ass}[thm]{Assumption}

\numberwithin{equation}{section}

\def\PI{P_{\rm I}}
\def\p{\partial}
\def\bbC{{\mathbb C}}

\newcommand{\Res}{\mathop{\rm Res}}

\begin{document}
\allowdisplaybreaks

\newcommand{\arXivNumber}{1507.06557}

\renewcommand{\PaperNumber}{011}

\FirstPageHeading

\ShortArticleName{Quantum Curve and the First Painlev\'e Equation}

\ArticleName{Quantum Curve and the First Painlev\'e Equation}

\Author{Kohei {IWAKI}~$^\dag$ and Axel {SAENZ}~$^\ddag$}

\AuthorNameForHeading{K.~Iwaki and A.~Saenz}

\Address{$^\dag$~Graduate School of Mathematics, Nagoya University, Nagoya, 464-8602, Japan}
\EmailD{\href{mailto:iwaki@math.nagoya-u.ac.jp}{iwaki@math.nagoya-u.ac.jp}}

\Address{$^\ddag$~Department of Mathematics, University of California, Davis, CA 95616-8633, USA}
\EmailD{\href{mailto:asaenz@math.ucdavis.edu}{asaenz@math.ucdavis.edu}}

\ArticleDates{Received August 04, 2015, in f\/inal form January 22, 2016; Published online January 29, 2016}

\Abstract{We show that the topological recursion for
the (semi-classical) spectral curve of the f\/irst Painlev\'e
equation $\PI$ gives a WKB solution for the isomonodromy problem
for~$\PI$. In other words, the isomonodromy system is a~quantum curve in the sense of
[Dumitrescu~O., Mulase~M., \textit{Lett. Math. Phys.} \textbf{104} (2014),
  635--671, arXiv:1310.6022] and [Dumitrescu~O., Mulase~M., arXiv:1411.1023].}


\Classification{34M55; 81T45; 34M60; 34M56}

\medskip

\rightline{\it Dedicated to Professor Takahiro Kawai on his seventieth birthday}

\section{Introduction}

{\em Painlev\'e transcendents} are remarkable
special functions which appear in many areas of
ma\-the\-matics and physics (e.g., \cite{FIKN}). These are solutions
of certain nonlinear ordinary dif\/ferential equations
known as {\em Painlev\'e equations}. These equations were
discovered by Painlev\'e and Gambier more than 100 years ago~\cite{Painleve}, and solutions have the so-called
{\em Painlev\'e property}; i.e., any movable singularity must be a pole.
One particular property of the Painlev\'e equations is
existence of the {\em Lax pair}; that is,
each Painlev\'e equation describes an {\em isomonodromic deformation}
of a certain meromorphic linear ordinary dif\/ferential equation~\cite{JMU, JM2}. The monodromy data of the linear ODEs
gives a conserved quantity of the Painlev\'e transcendents.
The {\it Riemann--Hilbert method}, as well as {\it exact WKB analysis}
are applied to analyze the properties of Painlev\'e transcendents~\cite{AKT-P, FIKN, Kap, KT-PI, Takei}.

On the other hand, {\em quantum curves} attract both mathematicians
and physicists since they are expected to encode
the information of many quantum topological invariants, such as
Gromov--Witten invariants, quantum knot invariants etc.
These are concieved in physics literature inclu\-ding~\cite{Agana, Agana2, DFM, GS}.
A quantum curve is an ordinary dif\/ferential (or dif\/ference)
equation containing a formal parameter $\hbar$
(which plays the role of the Planck constant),
like a Schr{\"o}dinger equation. The quantum invariants appear
in the coef\/f\/icients of the
{\em WKB $($Wentzel--Kramers--Brillouin$)$ solution}
of the quantum curve.

The Eynard--Orantin's {\em topological recursion}
introduced in \cite{EO07} is closely related to both of
the quantum curves and Painlev\'e equations (and many other topics).
Topological recursion is a~recursive algorithm to
compute the $1/N$-expansion of the correlation functions
and the partition function of matrix models from its
{\em spectral curve}, and it is generalized to any algebraic curve
which may not come from a matrix model.
In this context, quantum curves were f\/irst discussed in~\cite{BE}
for the Airy spectral curve, and generalized to spectral curves
with various backgrounds (see \cite{DM14, DM14-2, DBMNPS, GS, MS12}
and the survey article \cite{Nor}).
The spectral curves are recovered as the {\em semi-classical limit}
$\hbar \rightarrow 0$ of the quantum curves.
Moreover, the topological recursion is also closely related to
integrability~\cite{BBE12, Borot-Eynard, IM}
as is the relationship between matrix models
and integrable systems~\cite{DGZJ, Kont}.

The aim of this paper is to relate quantum curves and
the {\em first Painlev\'e equation} with a~formal parameter~$\hbar$
\begin{gather*}
\PI\colon \  \hbar^{2} \frac{d^{2}q}{dt^{2}} =
6q^{2} + t.
\end{gather*}
The (semi-classical) spectral curve for the isomonodormy
system associated with $\PI$ is given by
\begin{gather} \label{eq:sp-intro}
y^{2} = 4(x-q_{0})^{2}(x+2q_{0}),
\end{gather}
where $q_{0} = q_{0}(t)$ is an explicit function of~$t$.
This is a family of algebraic curves in $(x,y)$-space
parametrized by $t$. (The curve~\eqref{eq:sp-intro}
appeared in \cite[Section~10.6]{EO07} as the spectral curve of
(3,2)-minimal model.)
Our main result claims that,
starting from the spectral curve~\eqref{eq:sp-intro},
its {\em quantization} through the Eynard--Orantin's topological
recursion (in the sense of~\cite{DM14, DM14-2})
recovers the {\em whole} isomonodoromy system for $\PI$.

The precise statement of our main theorem is as follows.
Let $W_{g,n}(z_{1},\dots,z_{n})$ be the
{\em Eynard--Orantin differential} of type $(g,n)$
def\/ined from the spectral curve~\eqref{eq:sp-intro}
(see Section~\ref{section:top-recursion}).
These are meromorphic multi-dif\/ferential forms,
and $z_{i}$'s are copies of a coordinate on the
spectral curve~\eqref{eq:sp-intro}.
$W_{g,n}$'s also depend on $t$ since the
spectral curve depends on $t$.
Then, our main result states the following.

\begin{thm}[Theorem~\ref{conj:2}]
The following
WKB-type formal series $\psi(x,t, \hbar)$ defined by
\begin{gather} \label{eq:main-intro}
\psi(x(z),t,\hbar) := \exp\left( \sum_{g \ge 0,  n \ge 1}
\hbar^{2g-2+n}
\frac{1}{n!} \frac{1}{2^{n}} \int^{z}_{\bar{z}}\cdots\int^{z}_{\bar{z}}
W_{g,n}(z_{1},\dots,z_{n})
\right)
\end{gather}
satisfies the isomonodromy system associated with~$\PI$.
Here $x(z)$ is an explicit rational function of $z$ which appears
in the parametrization of the spectral curve~\eqref{eq:sp-intro},
and $\bar{z} = -z$.
\end{thm}

The above theorem tells us that the isomonodoromy system
associated with~$\PI$ is a quantum curve,
and its particular WKB solution is constructed
by the topological recursion as~\eqref{eq:main-intro}.
The main dif\/ferences between our theorem
and previous results on quantum curves are the following:
\begin{itemize}\itemsep=0pt
\item
Our quantum curve is a restriction of a certain
{\em partial differential equation} (a~holonomic system).
\item %
There are {\em infinitely many $\hbar$-correction terms}
in the quantum curve, and these correction terms
are essentially given by the {\em asymptotic expansion}
of the solution of $\PI$ for $\hbar \rightarrow 0$.
\end{itemize}

This paper is organized as follows.
In Section~\ref{section2}, we brief\/ly review some known facts about~$\PI$ together
with an important result on the WKB analysis of
isomonodromic systems developed by Kawai--Takei \cite{KT-PI, KT05}.
Our main theorem will be formulated in Section~\ref{section3}
after recalling the notion of topological recursion.
We will give a proof of the main results in Section~\ref{section:proof}.

\begin{rem}
After writing the draft version of this paper,
the authors were informed that B.~Eynard also
has the same result which has not been published yet,
but presented in~\cite{Eynard14}.
See also \cite[Chapter~5]{Eynard-book}.
\end{rem}

\section{The f\/irst Painlev\'e equation and isomonodromy system}\label{section2}

Let us consider the {\em first Painlev\'e equation}
with a formal parameter $\hbar$:
\begin{gather*}
\PI\colon \
\hbar^{2} \frac{d^{2}q}{dt^{2}} = 6q^{2} + t.
\end{gather*}
The equation $\PI$ is obtained from
\begin{gather*}
\frac{d^{2} \tilde{q}}{d\tilde{t}^{2}} =
6 \tilde{q}^{2} + \tilde{t}
\end{gather*}
via the rescaling
$\tilde{t} = \hbar^{-4/5} t$,
$\tilde{q} = \hbar^{-2/5} q$.
We will regard $\hbar$ as a small parameter (i.e., Planck's constant),
and investigate a particular formal solution
of $\PI$ which has an $\hbar$-expansion.

\subsection[Formal solution of $\PI$]{Formal solution of $\boldsymbol{\PI}$}

$\PI$ has the following formal power series solution:
\begin{gather} \label{eq:formal-q}
q(t,\hbar) = \sum_{n=0}^{\infty} \hbar^{2n}q_{2n}(t)
= q_{0}(t) + \hbar^{2} q_{2}(t)
+ \hbar^{4} q_{4}(t) + \cdots.
\end{gather}
It contains only even order terms of $\hbar$ since $\PI$
is invariant under $\hbar \mapsto - \hbar$.
The leading term $q_{0} = q_{0}(t)$ satisf\/ies
\begin{gather} \label{eq:q0}
6q_{0}^{2} + t = 0, \qquad \text{hence}~q_{0}(t)=\sqrt{-t/6},
\end{gather}
and the subleading terms are recursively determined by
\begin{gather} \label{eq:recursion-q-2k}
q_{2(k+1)}(t) = \frac{1}{12 q_{0}(t)} \left(
\frac{d^{2} q_{2k}}{dt^{2}}(t)  -
6 \sum_{k_{1}+k_{2} = k+1,\,  k_{i}>0} q_{2k_{1}}(t) q_{2k_{2}}(t)
\right), \qquad k \ge 1.
\end{gather}
As we will see, the coef\/f\/icients of the formal series appearing in
this paper are multivalued functions of $t$ and are def\/ined
on the Riemann surface of~$q_{0}$. Thus, in what follows,
we may use~$q_{0}$ instead of~$t$ when we express coef\/f\/icients.

The relation \eqref{eq:recursion-q-2k} implies
\begin{gather*}
q_{2k} =
c_{2k} q_{0}^{1-5k}, \qquad c_{2k} \in \bbC .
\end{gather*}
It is obvious that the coef\/f\/icients $q_{2k}(t)$ have
a singularity at $q_{0} = 0$ (i.e., $t = 0$).
This special point is called a {\em turning point} of $\PI$
\cite[Def\/inition 2.1]{KT-PI} (see also \cite[Section~4]{KT05}).
Throughout the paper, we assume the following:

\begin{ass}
The independent variable $t$ of $\PI$
lies on a domain that doesn't contain the origin.
\end{ass}

\begin{rem}
The formal solution \eqref{eq:formal-q} is called a
{\em $0$-parameter solution} of $\PI$ in \cite{KT05}
since it doesn't contain free parameters.
More general formal solutions having
one or two free para\-me\-ters (called 1- or 2-parameter solutions)
are constructed in~\cite{AKT-P} for all Painlev\'e equations
of second order.
See also~\cite{Umeta} for a construction of general formal
solutions of {\em higher} order Painlev\'e equations.
\end{rem}

\begin{rem}\sloppy
The formal solution \eqref{eq:formal-q} is in fact a~{\em divergent} series.
However, \cite[Theo\-rem~1.1]{Kamimoto-Koike} proved that
the formal solution is {\em Borel summable} when
$q_{0}$ satisf\/ies $q_{0} \ne 0$ and
\mbox{$\arg q_{0} \notin \{ \frac{2\ell}{5}\pi \,|\,
\ell \in {\mathbb Z} \}$}.
The exceptional set is called the {\em Stokes curve} of $\PI$.
(See \cite[Def\/ini\-tion~2.1]{KT-PI} for the notion of
Stokes curves of Painlev\'e equations
with a small parameter~$\hbar$.)
That is, there exists a function
which is analytic in $\hbar$ on a sectorial domain
with the center at the origin (which is also analytic in $t$)
such that~\eqref{eq:formal-q} is the asymptotic expansion
of the function for $\hbar \rightarrow 0$ in the sector.
The analytic function is called the {\em Borel sum}
of the formal series~\eqref{eq:formal-q}, and
it gives an analytic solution of $\PI$
(see~\cite{Costin08} for Borel summation method).
This particular asymptotic solution obtained by
the Borel summation method is called
the {\em tri-tronqu\'ee solution}
of $\PI$ (see~\cite{Joshi-Kitaev}), and the
{\em non-linear Stokes phenomena} on Stokes curves
are analyzed by \cite{FIKN, Kap, Takei}.
\end{rem}

\subsection[Isomonodromy system and the $\tau$-function]{Isomonodromy system and the $\boldsymbol{\tau}$-function}

It is known that $\PI$ describes the compatibility condition
for the following system of linear PDEs (cf.~\cite[Appendix C]{JM2}):
\begin{gather} \label{eq:Lax-matrix}
\hbar \frac{\p \Psi}{\p x}  = A \Psi, \qquad
\hbar \frac{\p \Psi}{\p t}  = B \Psi,
\end{gather}
where
\begin{gather*}
A   =
\begin{pmatrix}
A_{11} & A_{12} \\
A_{21} & A_{22}
\end{pmatrix}
:=
\begin{pmatrix}
  p
&
  4(x-q)
\\
\displaystyle x^{2} + q x + q^{2} + \frac{t}{2}
&
  - p
\end{pmatrix},
\\
B  =
\begin{pmatrix}
B_{11} & B_{12} \\
B_{21} & B_{22}
\end{pmatrix}
:=
\begin{pmatrix}
 0
&
   2
\\
\displaystyle \frac{x}{2} + q
&
  0
\end{pmatrix}.
\end{gather*}
The compatibility condition
\begin{gather*}
\hbar \frac{\p A}{\p t} - \hbar \frac{\p B}{\p x} + [A,B] = 0
\end{gather*}
is equivalent to the following Hamiltonian system
\begin{gather} \label{eq:Ham-PI}
\hbar \frac{dq}{dt} = \frac{\p H}{\p p}, \qquad
\hbar \frac{dp}{dt} = - \frac{\p H}{\p q},
\end{gather}
where the (time-dependent) Hamiltonian is given by
\begin{gather*}
H = H(q,p,t) := \frac{1}{2}p^{2} - 2 q^{3} - t q.
\end{gather*}
We can easily check that \eqref{eq:Ham-PI} and $\PI$ are equivalent.
The above system of linear ODEs
is called the {\em isomonodromy system} associated with $\PI$
(see~\cite{JMU, JM2}).

Let $(q,p) = (q(t,\hbar), p(t,\hbar))$ be a formal power series
solution of the Hamiltonian system~\eqref{eq:Ham-PI}; that is,
$q(t,\hbar)$ is the formal solution~\eqref{eq:formal-q}
of~$\PI$, and
\begin{gather*} 
p(t,\hbar) = \hbar \frac{d q(t,\hbar)}{dt} =
\sum_{n=0}^{\infty} \hbar^{2n+1} p_{2n+1}(t).
\end{gather*}
The corresponding Hamiltonian function is denoted by
\begin{gather} \label{eq:sigma}
\sigma(t,\hbar) := H\big(q(t,\hbar), p(t,\hbar), t\big).
\end{gather}
We can check that \eqref{eq:sigma}
is invariant under $\hbar \mapsto - \hbar$, and hence
it has the following expansion:
\begin{gather} \label{eq:parity-sigma}
\sigma(t,\hbar) = \sum_{n=0}^{\infty} \hbar^{2n} \sigma_{2n}(t).
\end{gather}

\begin{Def} [\cite{JMU, Okamoto}] \label{def:tau-function}
The {\em $\tau$-function}
(corresponding to the formal solution \eqref{eq:formal-q})
of $\PI$ is def\/ined by
\begin{gather} \label{eq:tau}
\hbar^{2} \frac{d}{dt} \log \tau(t,\hbar) = \sigma(t,\hbar)
\end{gather}
up to constant.
\end{Def}

The $\tau$-function can also be def\/ined in terms of
a solution of~\eqref{eq:Lax-matrix}~\cite{JMU} (see also
Appendix~\ref{section:relation-of-two-conjectures}).
The expansion~\eqref{eq:parity-sigma} implies that the $\tau$-function~\eqref{eq:tau} has an expansion of the form
\begin{gather*} 
\log \tau(t,\hbar) = \sum_{g=0}^{\infty} \hbar^{2g-2} \tau_{2g}(t).
\end{gather*} 

\subsection{Spectral curve}

In what follows, we assume that the formal solution
$(q(t,\hbar), p(t,\hbar))$ of~\eqref{eq:Ham-PI}
constructed above is substituted into the coef\/f\/icients
of the isomonodromy system \eqref{eq:Lax-matrix}.
Then, the coef\/f\/icients of the isomonodromy system has the following
$\hbar$-expansions:
\begin{gather*}
A   =   A_{0}(x,t) + \hbar A_{1}(x,t) + \hbar^{2} A_{2}(x,t) + \cdots, \\
B   =   B_{0}(x,t) + \hbar B_{1}(x,t) + \hbar^{2} B_{2}(x,t) + \cdots,
\end{gather*}
whose top terms are given by
\begin{gather*}
A_{0}(x,t)   =
\begin{pmatrix}
  0
&
 4(x-q_{0})
\\
\displaystyle x^{2} + q_{0} x + q_{0}^{2} + \frac{t}{2}
&
  0
\end{pmatrix}, \qquad
B_{0}(x,t)    =
\begin{pmatrix}
  0
&
  2
\\
\displaystyle \frac{x}{2}+q_{0}
&
   0
\end{pmatrix}.
\end{gather*}
Observe that, since $q_{0}$ satisf\/ies \eqref{eq:q0},
the algebraic curve def\/ined by
\begin{gather} \label{eq:sp-curve}
\det  (y - A_{0}(x,t)  ) =
y^{2} - 4  ( x-q_{0}  )^{2}  ( x+2q_{0}  )
= 0
\end{gather}
has {\em genus $0$}. Actually, this gives a family of algebraic
curves in ${\mathbb C}^{2}_{(x,y)}$
parametrized by $t$. Since we have assumed that
$t \ne 0$, $x=q_{0}$ and $x=-2q_{0}$ are distinct.

\begin{Def}
We call the algebraic curve~\eqref{eq:sp-curve}
{\em the semi-classical spectral curve}, or
{\em the spectral curve}
of (the f\/irst equation of) the isomonodromy system~\eqref{eq:Lax-matrix}.
\end{Def}

\begin{rem}
It is shown in \cite[Proposition~1.3]{KT-PI} that, for all (second order)
Painlev\'e equations with a~formal parameter $\hbar$,
the semi-classical spectral curves corresponding to
the same type of formal power series solution
as \eqref{eq:formal-q} have {\em genus $0$}.
\end{rem}

\begin{rem}
Since we are taking the semi-classical limit
(i.e., top term in $\hbar$-expansion),
our spectral curve \eqref{eq:sp-curve} is dif\/ferent from
usual spectral curves for isomonodromic deformation equations
discussed, e.g., in~\cite{Olsha, Takasaki}.
The spectral curves in the above papers have {\em higher genus}.
Recently, Nakamura~\cite{Nakamura} investigates
the geometry of genus~$2$ spectral curves which appear
in an {\em autonomous limit} of the
{\em 4th order Painlev\'e equations},
and use them to classify the Painlev\'e equations.
See \cite{KNS} for the list of 4th order Painlev\'e equations.
\end{rem}

\subsection{WKB analysis of isomonodromy system in scalar form}
\label{subsection:WKB-for-scalar-Lax}

Denote the unknown vector function of \eqref{eq:Lax-matrix}
by $\Psi = {}^t(\psi_1, \psi_2)$.
Then, $\psi = \psi_1$ satisf\/ies the
following scalar version of isomonodromy system
\begin{gather}
\left( \left(\hbar\frac{\p}{\p x}\right)^{2} +
f \left(\hbar\frac{\p}{\p x}\right) + g \right)\psi = 0, \nonumber\\
\hbar \frac{\p \psi}{\p t} = \frac{1}{2(x-q)} \left(
\hbar \frac{\p \psi}{\p x} - p \psi \right),\label{eq:Lax-scalar}
\end{gather}
where
\begin{gather*}
f = f(x,t,\hbar)   :=   -\operatorname{tr} A - \hbar \frac{\p}{\p x}\log A_{12}
= -\hbar \frac{1}{x-q},  \\
g = g(x,t,\hbar)   :=   \det A - \hbar \frac{\p A_{11}}{\p x} +
\hbar A_{11} \frac{\p}{\p x}\log A_{12}   \\
\hphantom{g}{}
  =
- \big(4x^{3} + 2tx + p^{2} - 4q^{3} - 2tq\big) + \hbar \frac{p}{x-q}.
\end{gather*}
The coef\/f\/icients of $f$ and $g$ have an $\hbar$-expansion since
$q$ and $p$ are contained in them
\begin{gather} \label{eq:hbar-correction-in-quantization}
f   =   - \hbar  \frac{1}{x-q_{0}}
 +\hbar^{3} \frac{1}{1728q_{0}^{4}(x-q_{0})^{2}} + \hbar^{5}
 \frac{49x-51q_{0}}{5971968q_{0}^{9}(x-q_{0})^{3}} + \cdots,
\\  \label{eq:hbar-correction-in-quantization2}
g   =   -4(x-q_{0})^{2}(x+2q_{0})
 - \hbar^{2} \frac{x+11q_{0}}{144q_{0}^{2}(x-q_{0})}
 - \hbar^{4} \frac{7x^{2}+34q_{0}x-53q_{0}^{2}}
  {248832q_{0}^{7}(x-q_{0})^{2}}+\cdots .
\end{gather}
The top term of $g$ appears in the def\/ining equation of
the spectral cur\-ve~\eqref{eq:sp-curve},
and its zeros are called {\em turning points}
of the f\/irst equation of~\eqref{eq:Lax-scalar}
in the WKB analysis.
In particular, under the assumption $t \ne 0$, there is
\begin{itemize}\itemsep=0pt
\item
a {\em simple} turning point at $x=-2q_{0}$
which is a branch point
of the spectral curve \eqref{eq:sp-curve}, and
\item
a {\em double} turning point at $x = q_{0}$
which is a singular point of
the spectral curve \eqref{eq:sp-curve}.
\end{itemize}
Consider the {\em Riccati equation}
\begin{gather} \label{eq:Riccati}
\hbar^{2} \left( P^{2} + \frac{\p P}{\p x} \right)
+ f \hbar P + g = 0.
\end{gather}
This is equivalent to the f\/irst equation
in \eqref{eq:Lax-scalar} by
\begin{gather*}
\psi = \exp\left( \int^{x} Pdx \right), \qquad
 \text{i.e.}, \quad P = \frac{1}{\psi} \frac{\p \psi}{\p x}  .
\end{gather*}
Let
\begin{gather*}
P^{(\pm)}(x,t,\hbar) = \sum_{m=0}^{\infty}\hbar^{m-1}P^{(\pm)}_{m}(x,t)
\end{gather*}
be the formal solutions of \eqref{eq:Riccati} with the top term
\begin{gather*} 
P^{(\pm)}_{0}(x,t) = \pm 2(x-q_{0})\sqrt{x+2q_{0}}.
\end{gather*}
The coef\/f\/icients $P^{(\pm)}_{m}(x,t)$ are recursively determined by
\begin{gather} \label{eq:recursion-for-Pm}
2P^{(\pm)}_{0}P^{(\pm)}_{m+1} +
\sum_{\substack{a+b=m+1 \\ a,b \ge 1}}
P^{(\pm)}_{a} \big(P^{(\pm)}_{b} + f_{b}\big) +
\frac{\p P^{(\pm)}_{m}}{\p x} + g_{m+1} = 0 \qquad \text{for} \quad m \ge 0,
\end{gather}
where $f_{a}$ and $g_{a}$ are the coef\/f\/icient of $\hbar^{a}$
in $f$ and $g$, respectively. Explicit forms of the f\/irst few terms
are given by
\begin{gather*}
P^{(\pm)}_{1}   =   -\frac{1}{4(x+2q_{0})}, \qquad
P^{(\pm)}_{2}   =   \pm \frac{x+17q_{0}}{576q_{0}^{2}(x+2q_{0})^{5/2}}, \qquad
P^{(\pm)}_{3}   =   -\frac{2x^{2}+20q_{0}x+77q_{0}^{2}}
{6912q_{0}^{4}(x+2q_{0})^{4}}, \\
P^{(\pm)}_{4}   =   \pm \frac{28x^{4}+500q_{0}x^{3}+
3684q_{0}^{2}x^{2}+14273q_{0}^{3}x+27307q_{0}^{4}}
{3981312 q_{0}^{7}(x+2q_{0})^{11/2}}.
\end{gather*}
It is obvious from~\eqref{eq:recursion-for-Pm}
that $P^{(\pm)}_{m}(x,t)$ are holomorphic except
at the turning points and $x = \infty$
(and multivalued for even $m$).
It also follows from the recursion relation~\eqref{eq:recursion-for-Pm} that
\begin{gather} \label{eq:asymp-P}
P^{(\pm)}(x,t,\hbar)
= \pm \left(\frac{2}{\hbar}x^{3/2} + \frac{t}{2\hbar}x^{-1/2}
\mp \frac{1}{4}x^{-1} + \frac{\sigma(t,\hbar)}{2\hbar} x^{-3/2}
+ O\big(x^{-2}\big) \right)
\end{gather}
holds when $x \rightarrow \infty$.

\begin{rem} \label{rem:asymptotics}
We can check that $P^{(\pm)}_{m}(x,t)$'s
have the following asymptotic expansion for \mbox{$x \rightarrow \infty$}
\begin{gather*}
P^{(\pm)}_{0}(x,t)   =   \pm \left( 2x^{3/2} + \frac{t}{2}x^{-1/2}
+ O\big(x^{-3/2}\big) \right), \qquad
P^{(\pm)}_{1}(x,t)  =   -\frac{1}{4}x^{-1} + O\big(x^{-3/2}\big), \\
P^{(\pm)}_{m}(x,t)   =    O\big(x^{-3/2}\big) \qquad \text{for} \quad m \ge 2,
\end{gather*}
and we have \eqref{eq:asymp-P} after summing up
$\hbar^{m-1}P^{(\pm)}_{m}(x,t)$. Once you know that $P^{(\pm)}(x,t,\hbar)$
has an asymptotic expansion in this sense, subleading terms in
\eqref{eq:asymp-P} can be computed from
the Riccati equation \eqref{eq:Riccati}.
\end{rem}

Def\/ine
\begin{gather*}
P_{\rm odd}(x,t,\hbar)   :=   \frac{1}{2}
\big( P^{(+)}(x,t,\hbar) - P^{(-)}(x,t,\hbar) \big), \\
P_{\rm even}(x,t,\hbar)   :=   \frac{1}{2}
\big( P^{(+)}(x,t,\hbar) + P^{(-)}(x,t,\hbar) \big).
\end{gather*}
It is easy to check that (cf.~\cite[Section~2]{KT05})
\begin{gather} \label{eq:odd-even-relation}
P_{\rm even}(x,t,\hbar) =
-\frac{1}{2}\frac{\p}{\p x}
\log \frac{ \hbar P_{\rm odd}(x,t,\hbar)}
{2(x-q(t,\hbar))}
\end{gather}
and
\begin{gather} \label{eq:parity}
P_{\rm odd}(\sigma(x),t,\hbar) = - P_{\rm odd}(x,t,\hbar)
\end{gather}
hold. Here $x$ is regarded as a coordinate on the spectral curve,
and $\sigma$ is the covering involution for the spectral curve:
$P^{(+)}(\sigma(x),t) = - P^{(-)}(x,t)$.

Since
\begin{gather*}
\frac{ \hbar P_{\rm odd}(x,t,\hbar)}{2(x-q(t,\hbar))} =
\sqrt{x+2q_{0}} \big( 1 + O(\hbar) \big),
\end{gather*}
the right hand-side of~\eqref{eq:odd-even-relation}
is the derivative of the formal power series
\begin{gather*}
-\frac{1}{2}\log\frac{ \hbar P_{\rm odd}(x,t,\hbar)}{2(x-q(t,\hbar))} =
-\frac{1}{4}\log(x+2q_{0}) + O(\hbar).
\end{gather*}
Thus the ambiguity of the branch of the logarithm
only appears in the top term, but
we care about the ambiguity since
it doesn't matter in our computation.

The following theorem was applied in the
{\it transformation theory of Painlev\'e equations}
in~\cite{KT-PI}. We will use the fact
in the proof of our main theorem.

\begin{thm}[\protect{cf.~\cite[Proposition~1.2 and Theorem~1.1]{KT-PI}}]
\label{thm:Podd-and-t-derivative} \quad
\begin{itemize}\itemsep=0pt
\item[$(i)$]
The formal series $P^{(\pm)}(x,t,\hbar)$ satisfies
\begin{gather} \label{eq:dt-P}
\hbar \frac{\p}{\p t}P^{(\pm)}(x,t,\hbar) = \frac{\p}{\p x}\left(
\frac{\hbar P^{(\pm)}(x,t,\hbar) - p(t,\hbar)}{2(x-q(t,\hbar))}
\right).
\end{gather}
In particular, $P_{\rm odd}(x,t,\hbar)$ satisfies
\begin{gather} \label{eq:dt-Podd}
\frac{\p}{\p t}P_{\rm odd}(x,t,\hbar) =
\frac{\p}{\p x}\left( \frac{P_{\rm odd}(x,t,\hbar)}{2(x-q(t,\hbar))}
\right).
\end{gather}

\item[$(ii)$]
All coefficients of $P^{(\pm)}(x,t,\hbar)$ are holomorphic
except at the simple turning point $x=-2q_{0}$ and $x= \infty$.
In particular, they are holomorphic at
the double turning point $x=q_{0}$.

\item[$(iii)$]
The formal series
\begin{gather}
\psi_{\pm}(x,t,\hbar)   :=    \exp\left(
\pm \int^{x}_{v} P_{\rm odd}(x',t,\hbar)dx'
- \frac{1}{2} \log \frac{ \hbar P_{\rm odd}(x,t,\hbar)}
{2(x-q(t,\hbar))}\right)  \nonumber \\
 \hphantom{\psi_{\pm}(x,t,\hbar)}{} =
\left(\frac{2(x-q(t,\hbar))}{\hbar P_{\rm odd}(x,t,\hbar)} \right)^{1/2}
\exp\left( \pm \int^{x}_{v} P_{\rm odd}(x',t,\hbar) dx' \right)\label{eq:WKB-IM}
\end{gather}
satisfies the isomonodoromy system~\eqref{eq:Lax-scalar}.
Here $v$ is the simple turning point $-2q_{0}$.
The integral from $v$ is defined by
\begin{gather} \label{eq:contour-integral-Podd}
\int^{x}_{v} P_{\rm odd}(x',t,\hbar) dx' =
\frac{1}{2} \int_{\gamma_x} P_{\rm odd}(x',t,\hbar) dx',
\end{gather}
where the path $\gamma_x$ is depicted in Fig.~{\rm \ref{fig:contour}}
$($cf.~{\rm \cite[Section~2]{KT05})}.
\end{itemize}
\end{thm}

\begin{figure}[t]\centering
 \includegraphics{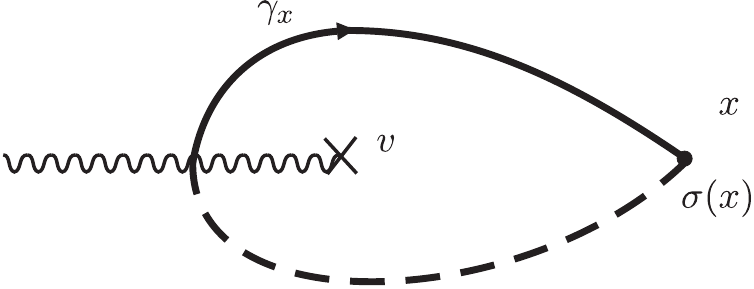}
 \caption{For a~given~$x$, the path~$\gamma_{x}$ starts from
 the point~$\sigma(x)$ and ends at~$x$.
 The wiggly lines designate a branch cut, and the solid (resp.\ dotted) part
 represents a~part of path on the f\/irst (resp.\ the second)
 sheet of the spectral curve.}
 \label{fig:contour}
 \end{figure}

\begin{proof}
Although the scalar version of isomonodromy system \eqref{eq:Lax-scalar}
is dif\/ferent from that used in~\cite{KT-PI}, they are
related by a gauge transformation $\psi \mapsto (x-q)^{1/2} \psi$.
Therefore, the equalities~\eqref{eq:dt-P} and~\eqref{eq:dt-Podd}
in~(i) together with the holomorphicity of each coef\/f\/icient
of $P_{\rm odd}(x,t,\hbar)$ at $x=q_{0}$ follows from
\cite[Proposition~1.2 and Theorem~1.1]{KT-PI}.
Then, it turns out that the coef\/f\/icients of
$P_{\rm odd}(x,t,\hbar)/(x-q(t,\hbar))$
are also holomorphic due to~\eqref{eq:dt-Podd}.
Then, \eqref{eq:odd-even-relation}
implies that each coef\/f\/icient of $P_{\rm even}(x,t,\hbar)$ is also
holomorphic at $x=q_{0}$.
Thus we have proved~(ii).

The claim (iii) follows from a straightforward computations
\begin{gather*}
\frac{1}{\psi_{\pm}} \frac{\p \psi_{\pm}}{\p x}  =
\pm P_{\rm odd} -\frac{1}{2}\frac{\p}{\p x} \log
\left(\frac{ \hbar P_{\rm odd}}{2(x-q)}\right)
= \pm P_{\rm odd} + P_{\rm even}
= P^{(\pm)}, \\
\hbar\frac{1}{\psi_{\pm}} \frac{\p \psi_{\pm}}{\p t}   =
\frac{1}{2}\left( \frac{-\hbar (dq/dt)}{x-q} - \frac{\hbar}{P_{\rm odd}}
\frac{\p P_{\rm odd}}{\p t} \right) \pm \int^{x}_{v}
\hbar\frac{\p P_{\rm odd}}{\p t} dx \\
\hphantom{\hbar\frac{1}{\psi_{\pm}} \frac{\p \psi_{\pm}}{\p t}  }{}
  =
-\frac{p}{2(x-q)} - \frac{\hbar}{P_{\rm odd}}
\left( \frac{1}{2(x-q)}\frac{\p P_{\rm odd}}{\p x} -
\frac{P_{\rm odd}}{2(x-q)^{2}} \right) \pm \hbar \frac{P_{\rm odd}}{2(x-q)} \\
\hphantom{\hbar\frac{1}{\psi_{\pm}} \frac{\p \psi_{\pm}}{\p t}  }{} =
-\frac{p}{2(x-q)} + \frac{\hbar P^{(\pm)}}{2(x-q)} =
\frac{1}{2(x-q)} \left( \frac{\hbar}{\psi_{\pm}}
\frac{\p \psi_{\pm}}{\p x} - p \right).\tag*{\qed}
\end{gather*}
\renewcommand{\qed}{}
\end{proof}

As we will see below, an {\em isomonodromic WKB solution}
such as~\eqref{eq:WKB-IM}
is constructed from just a family of algebraic curves~\eqref{eq:sp-curve}
by {\em the topological recursion}~(\cite{EO07}).
In particular, the f\/irst equation
in \eqref{eq:Lax-scalar} gives a {\em quantization}
of the spectral curve~\eqref{eq:sp-curve}
in the sense of~\cite{DM14, DM14-2}.

\begin{rem}
In the above computation the normalization \eqref{eq:WKB-IM}
is essential. Since $P_{\rm odd}$ is anti-invariant under the
covering involution $\sigma$ as \eqref{eq:parity} and
the integral in~\eqref{eq:WKB-IM} is
def\/ined as a contour integral~\eqref{eq:contour-integral-Podd},
we don't need to take care of the branch point~$v$ in the computation
\begin{gather*}
\int^{x}_{v} \frac{\partial P_{\rm odd}}{\partial t}dx =
\frac{1}{2}\int^{x}_{\sigma(x)}
\frac{\p}{\p x}\left( \frac{P_{\rm odd}}{2(x-q)}\right)dx
= \frac{P_{\rm odd}}{2(x-q)}.
\end{gather*}
\end{rem}

\begin{rem}
We can also construct a WKB-type formal solution
of matrix isomonodromy system \eqref{eq:Lax-matrix}.
Def\/ine
\begin{gather*}
\tilde{\psi}_{\pm}(x,t,\hbar)   =
\frac{ \hbar \frac{d \psi_{\pm}}{dx}(x,t,\hbar) - A_{11}(x,t,\hbar)
\psi_{\pm}(x,t,\hbar)}{A_{12}(x,t,\hbar)}  \\
\hphantom{\tilde{\psi}_{\pm}(x,t,\hbar) }{}
  =
\frac{\hbar P^{(\pm)}(x,t,\hbar) - A_{11}(x,t,\hbar)}
{A_{12}(x,t,\hbar)} \psi_{\pm}(x,t,\hbar).
\end{gather*}
Then, the matrix valued formal series
\begin{gather} \label{eq:WKB-IM-mat}
\Psi(x,t,\hbar)   =
\begin{pmatrix}
\psi_{+}(x,t,\hbar) & \psi_{-}(x,t,\hbar) \\[+.7em]
\tilde{\psi}_{+}(x,t,\hbar) & \tilde{\psi}_{-}(x,t,\hbar)
\end{pmatrix}
\end{gather}
gives a fundamental formal solution of the
isomonodoromy system~\eqref{eq:Lax-matrix}.
\end{rem}

\section{Topological recursion and quantum curve theorem}\label{section3}

In this section we review the Eynard--Orantin's
topological recursion~\cite{EO07} for our
spectral cur\-ve~\eqref{eq:sp-curve},
and formulate our main theorem.

\subsection{Topological recursion}
\label{section:top-recursion}

The topological recursion is an algorithm associating some
dif\/ferential forms $W_{g,n}$ and num\-bers~$F_{g}$
given the following source data:
\begin{itemize}\itemsep=0pt
\item
A plane curve $({\mathcal C}, x, y)$: ${\mathcal C}$ is a
compact Riemann surface, $x, y\colon {\mathcal C} \rightarrow {\mathbb P}^{1}$
are meromorphic functions.
\item
The Bergman kernel $B$:
It is a symmetric dif\/ferential form on ${\mathcal C} \times {\mathcal C}$
with poles of order 2 along the diagonal, and
satisfying some normalization conditions.
\end{itemize}

In our case, ${\mathcal C} = {\mathbb P}^1$ and $x, y$ are
rational functions which parametrize the spectral curve~\eqref{eq:sp-curve}
\begin{gather} \label{eq:par-rep}
x(z) = z^{2} - 2 q_{0}, \qquad
y(z) = 2z(z^{2} - 3q_{0}).
\end{gather}
Here $z$ is a coordinate on ${\mathbb P}^1$.
The Bergman kernel is given by
\begin{gather*}
B(z_1, z_2) = \frac{dz_{1}dz_{2}}{(z_{1}-z_{2})^{2}},
\end{gather*}
since the spectral curve is of genus $0$.
Zeros of~$dx$ are called {\em ramification points} of the
spectral curve~\eqref{eq:par-rep}.
Our spectral curve has only one ramif\/ication point at $z=0$.

The topological recursion for our spectral curve~\eqref{eq:par-rep}
is formulated as follows (see~\cite{EO07} for general case):

\begin{Def}[\protect{\cite[Def\/inition~4.2]{EO07}  (see also~\cite[Section~3]{DM14})}]
The {\em Eynard--Orantin differential}
$W_{g,n}(z_{1},\dots,z_{n})$ of type $(g,n)$
is a meromorphic $n$-dif\/ferential on the $n$-times product
of the spectral curve~\eqref{eq:par-rep} def\/ined by
the following {\em topological recursion relation}:
\begin{itemize}\itemsep=0pt
\item
for $2g-2+n \le 0$:
\begin{gather*}
W_{0,1}(z_{1})   :=   y(z_{1}) dx(z_{1})  = 4z_{1}^{2} \big(z_{1}^{2}-3q_{0}\big)dz_{1}  , \\
W_{0,2}(z_{1},z_{2})   :=   B(z_1, z_2)   = \frac{dz_{1}dz_{2}}{(z_{1}-z_{2})^{2}} ,
\end{gather*}
\item %
for $2g-2+n = 1$:
\begin{gather*}
W_{0,3}(z_{1},z_{2},z_{3})   :=   \frac{1}{2\pi i}
\oint_{\gamma_{0}} K(z,z_{1}) \big[ W_{0,2}(z,z_{2}) W_{0,2}(\bar{z},z_{3})
+ W_{0,2}(z,z_{3}) W_{0,2}(\bar{z},z_{2}) \big], \\
W_{1,1}(z_{1})   :=   \frac{1}{2\pi i}
\oint_{\gamma_{0}} K(z,z_{1}) W_{0,2}(z,\bar{z}),
\end{gather*}
\item %
for $2g-2+n \ge 2$:
\begin{gather}
W_{g,n}(z_{1},\dots,z_{n}) := \frac{1}{2\pi i}
\oint_{\gamma_{0}} K(z,z_{1})   \nonumber \\
 \qquad{}\times  \Biggl[ \sum_{j=2}^{n}
\big( W_{0,2}(z,z_{j}) W_{g,n-1}\big(\bar{z}, z_{[\hat{1},\hat{j}]}\big) +
W_{0,2}(\bar{z},z_{j}) W_{g,n-1}\big(z, z_{[\hat{1},\hat{j}]}\big) \big) \nonumber \\
\qquad{}+ W_{g-1,n+1}\big(z,\bar{z},z_{[\hat{1}]}\big) +
\sum_{\substack{g_{1}+g_{2}=g \\ I \sqcup J = [\hat{1}]}}^{\text{stable}}
W_{g_{1}, |I|+1}(z,z_{I}) W_{g_{2}, |J|+1}(\bar{z}, z_{J}) \Biggr].\label{eq:top-rec}
\end{gather}
\end{itemize}
Here $\gamma_{0}$ is a small cycle (in $z$-plane) which encircles
the ramif\/ication point $z=0$ in the counter-clockwise direction,
$\bar{z} = - z$ is the conjugate of $z$ near the ramif\/ication point,
and the {\em recursion kernel}~$K(z,z_{1})$ is given by
\begin{gather*} 
K(z,z_{1}) = - \frac{\omega^{\bar{z}-z}(z_{1})}{2(y(z) - y(\bar{z}))dx(z)}, \qquad
\omega^{\bar{z}-z}(z_{1}) = \int^{\bar{z}}_{z} W_{0,2}( \cdot, z_{1}).
\end{gather*}
Also, we use the index convention
$[\hat{j}] = \{1, \dots, n \}{\setminus}\{j\}$ and so on.
Lastly, the sum in the third line of~\eqref{eq:top-rec} is taken
for indices in the stable range
(i.e., only $W_{g,n}$'s with $2g - 2 + n \ge 1$ appear).
\end{Def}

The explicit form of some of Eynard--Orantin dif\/ferentials
are given as follows
\begin{gather*}
W_{0,3}   =   \frac{1}{12q_{0}z_{1}^{2}z_{2}^{2}z_{3}^{2}}
dz_{1}dz_{2}dz_{3}, \\
W_{0,4}   =   \frac{z_{1}^{2}z_{2}^{2}z_{3}^{2}z_{4}^{2} +
3q_{0} (z_{1}^{2}z_{2}^{2}z_{3}^{2}+z_{2}^{2}z_{3}^{2}z_{4}^{2}+
z_{3}^{2}z_{4}^{2}z_{1}^{2}+
z_{4}^{2}z_{1}^{2}z_{2}^{2})}{144q_{0}^{3}
z_{1}^{4}z_{2}^{4}z_{3}^{4}z_{4}^{4}} dz_{1}dz_{2}dz_{3}dz_{4}, \\
W_{1,1}   =   \frac{z_{1}^{2} + 3q_{0}}
{288 q_{0}^{2} z_{1}^{4}} dz_{1},
\\
W_{1,2}   =   \frac{ 2z_{1}^{4}z_{2}^{4} +
6q_{0}(z_{1}^{4}z_{2}^{2} + z_{1}^{2} z_{2}^{4}) +
3q_{0}^{2} (5z_{1}^{4} + 3 z_{1}^{2}z_{2}^{2} + 5z_{2}^{4})  }
{3456 q_{0}^{4} z_{1}^{6}z_{2}^{6}} dz_{1}dz_{2}, \\
W_{2,1}   =   \frac{ 28 z_{1}^{8} + 84q_{0} z_{1}^{6} + 252q_{0}^{2}z_{1}^{4}
+ 609q_{0}^{3}z_{1}^{2}  + 945 q_{0}^{4}}
{1990656q_{0}^{7} z_{1}^{10}}dz_{1}.
\end{gather*}

Eynard--Orantin dif\/ferentials have
the following properties (see~\cite{EO07}):
\begin{itemize}\itemsep=0pt
\item
As a dif\/ferential form on each variable $z_{i}$,
$W_{g,n}$, for $2g-2+n \ge 1$, is {\em holomorphic}
except for the ramif\/ication point $0$ and may
have a pole at $0$.
\item
$W_{g,n}$ is {\em symmetric}; that is, they are invariant
under any permutation of variables.
\item
For $2g-2+n \ge 1$, $W_{g,n}$
is {\em anti-invariant} under the involution
$z_{i} \mapsto \bar{z}_{i}$ for each variable:
\begin{gather*} 
W_{g,n}(z_{1},\dots,\bar{z}_{j},\dots,z_{n}) = -
W_{g,n}(z_{1},\dots,z_{j},\dots,z_{n}) \qquad
\text{for} \quad j = 1,\dots,n.
\end{gather*}

\item
$W_{g,n}$ is also {\em holomorphic} in $t$ except for
$t = 0$ (i.e., $q_{0} = 0$).
There is a formula for the derivative of $W_{g,n}$
with respect to $t$; see Section~\ref{section:variation-formula}.
\end{itemize}

\subsection{Quantum curve theorem}
\label{section:quantum-curve-conjecture}

In this section we describe our main result which claims
that the scalar isomonodromy system~\eqref{eq:Lax-scalar}
gives a {\em quantum curve}.

\begin{Def}
For $g \ge 0, n \ge 1$ satisfying $2g-2+n \ge 1$,
def\/ine {\em open free energy} of type ($g,n$) by
\begin{gather} \label{eq:open-Fgn}
F_{g,n}(z_{1},\dots,z_{n}) :=
\frac{1}{2^{n}} \int^{z_{1}}_{\bar{z}_{1}}\cdots
\int^{z_{n}}_{\bar{z}_{n}} W_{g,n}(z_{1},\dots,z_{n}).
\end{gather}
\end{Def}

It follows from the def\/inition that open free energies satisfy
\begin{gather*}
d_{z_{1}}\cdots d_{z_{n}} F_{g,n}(z_{1},\dots,z_{n})   =
W_{g,n}(z_{1},\dots,z_{n}), \\
F_{g,n}(z_{1},\dots,\bar{z}_{j}, \dots,z_{n})   =
- F_{g,n}(z_{1},\dots, z_{j}, \dots,z_{n})
\qquad \text{for} \quad j=1,\dots,n.
\end{gather*}
Explicit computation shows that
\begin{gather*}
F_{0,3}(z_{1},z_{2},z_{3})   =   -\frac{1}{12q_{0}z_{1}z_{2}z_{3}}, \\
F_{0,4}(z_{1},z_{2},z_{3},z_{4})   =
\frac{z_{1}^{2}z_{2}^{2}z_{3}^{2}z_{4}^{2} + q_{0}\bigl(
z_{1}^{2}z_{2}^{2}z_{3}^{2} + z_{2}^{2}z_{3}^{2}z_{4}^{2} +
z_{3}^{2}z_{4}^{2}z_{1}^{2} + z_{4}^{2}z_{1}^{2}z_{2}^{2} \bigr)}
{144q_{0}^{3}z_{1}^{3}z_{2}^{3}z_{3}^{3}z_{4}^{3}},
\\
F_{1,1}(z_{1})  =  -\frac{z_{1}^{2}+q_{0}}{288q_{0}^{2} z_{1}^{3}}, \\
F_{1,2}(z_{1},z_{2}) =
\frac{2z_{1}^{4}z_{2}^{4} +
2q_{0}\bigl( z_{1}^{4}z_{2}^{2} + z_{1}^{2}z_{2}^{4} \bigr) +
q_{0}^{2}\bigl( 3z_{1}^{4} + z_{1}^{2}z_{2}^{2} + 3z_{2}^{4} \bigr)}
{3456q_{0}^{4}z_{1}^{5}z_{2}^{5}}, \\
F_{2,1}(z_{1})  =
-\frac{140z_{1}^{8} + 140 q_{0} z_{1}^6 + 252 q_{0}^{2} z_{1}^4
+ 435 q_{0}^{3} z_{1}^2 + 525 q_{0}^{4}}
{9953280 q_{0}^{7} z_{1}^{9}}.
\end{gather*}
We also introduce functions $\{ S_{m}(x,t) \}_{m \ge 0}$ by
\begin{gather*} 
S_{0}(x,t) := \int^{x}_{v} y(z(x'))dx', \qquad
S_{1}(x,t) := -\frac{1}{2} \log \left(\frac{y(z(x))}{2(x-q_{0})}
\right),
\end{gather*}
and for $m \ge 2$
\begin{gather*} 
S_{m}(x,t) := \sum_{\substack{2g-2+n=m-1 \\ g\ge0, n\ge 1}}
\frac{F_{g,n}(z,\dots,z)}{n!} \biggr|_{z=z(x)},
\end{gather*}
where $z(x) = \sqrt{x+2q_{0}}$ is the inverse function of $x(z)$.
After computations we have
\begin{gather*}
S_{0}(x,t)  =  \frac{4}{5}(x-3q_{0})(x+2q_{0})^{3/2}, \qquad
S_{1}(x,t)  =  - \frac{1}{4}\log(x+2q_{0}), \\
S_{2}(x,t)  =  -\frac{x+7q_{0}}{288q_{0}^{2}(x+2q_{0})^{3/2}}, \qquad
S_{3}(x,t)  =  \frac{2x^{2}+14q_{0}x+35q_{0}^{2}}
{6912q_{0}^{4}(x+2q_{0})^{3}}, \\
S_{4}(x,t)  =  -\frac{140x^{4}+1580q_{0}x^{3}
+ 7476 q_{0}^{2}x^{2}+18739q_{0}^{3}x+23499q_{0}^{4}}
{9953280 q_{0}^{7}(x+2q_{0})^{9/2}}.
\end{gather*}

Our main result is the following.

\begin{thm} \label{conj:2}
The formal series $\psi(x,t,\hbar)$ given by
\begin{gather}
\psi(x,t,\hbar)   :=    \exp  ( S(x,t,\hbar)  t),
\label{eq:wave-function} \\
S(x,t,\hbar)   :=   \sum_{m = 0}^{\infty} \hbar^{m-1} S_{m}(x,t)
\label{eq:Riccati-solution}
\end{gather}
satisfies both of the differential equations in
scalar-version of the isomonodromy system~\eqref{eq:Lax-scalar}.
That is, the formal series $S(x,t,\hbar)$ given by~\eqref{eq:Riccati-solution} satisfies
the following differential equations which are equivalent to~\eqref{eq:Lax-scalar}:
\begin{gather}  \label{eq:Riccati-2}
 \hbar^{2}\left( \left(\frac{\p S}{\p x}\right)^{2} +
\frac{\p^{2} S}{\p x^{2}} \right)   =
\frac{\hbar}{x-q} \left( \hbar \frac{\p S}{\p x} - p \right)
+ \big(4x^{3}+2t x + p^{2} - 4q^{3}-2tq\big),   \\
\label{eq:dt-Riccati-2}
 \hbar\frac{\p S}{\p t}   =   \frac{1}{2(x-q)} \left(
\hbar \frac{\p S}{\p x} - p \right).
\end{gather}
\end{thm}

Thus, the {\em principal specialization}
(i.e., setting $z_{i} = z$ for all $i = 1,\dots,n$)
of the open free energies gives an isomonodromic WKB solution.
Theorem~\ref{conj:2} implies
\begin{gather} \label{eq:dS-is-P}
\frac{\p S}{\p x}(x,t,\hbar) = P^{(+)}(x,t,\hbar)
\end{gather}
holds (under a suitable choice of the branch of $\sqrt{x+2q_{0}}$).
The computational results in Section~\ref{subsection:WKB-for-scalar-Lax}
show that~\eqref{eq:dS-is-P}
holds up to~$\hbar^{4}$. A~full proof of Theorem~\ref{conj:2}
will be given in Section~\ref{section:proof} together
with that of Theorem~\ref{conj:1} below.

\begin{rem}
In the topological recursion~\eqref{eq:top-rec},
we take residues only at the ramif\/ication point $z = 0$.
Thus $W_{g,n}$'s def\/ined here
are dif\/ferent from those in~\cite{DM14-2};
in particular, our quantum curve~\eqref{eq:Lax-scalar}
has {\em infinitely many $\hbar$-corrections} as in~\eqref{eq:hbar-correction-in-quantization} and~\eqref{eq:hbar-correction-in-quantization2}
(but recovers the same spectral curve in the semi-classical limit).
\end{rem}

\begin{rem}
In Theorem~\ref{conj:2}, the choice of the lower end points
of the integral in~\eqref{eq:open-Fgn} is important.
Dif\/ferent choice also give a WKB solution of the f\/irst equation
in~\eqref{eq:Lax-scalar}, but it may not satisfy
the second equation in general.
\end{rem}

\subsection[Closed free energies and the $\tau$-function]{Closed free energies and the $\boldsymbol{\tau}$-function}

The other main result of this paper is giving
another proof of the known fact
about the relationship between the closed free energies
and the $\tau$-function of~$\PI$
(cf.~\cite{DGZJ,EO07}).

\begin{Def}[\protect{\cite[Def\/inition~4.3]{EO07}}]
Def\/ine the {\em closed free energy}
$F_{g} = F_{g}(t)$ for $g \ge 2$ by
\begin{gather*}
F_{g}(t) = \frac{1}{2\pi i (2-2g)}
\oint_{\gamma_{0}} \Phi(z) W_{g,1}(z),
\end{gather*}
where
\begin{gather*}
\Phi(z) = \int^{z}_{z_{0}} y(z)dx(z)
  = \frac{4}{5}z^{5} - 4q_{0} z^{3} +
\text{const}
\end{gather*}
and $z_{0}$ is a generic point.
Free energies $F_{0}$ and $F_{1}$ for $g = 0, 1$
are also def\/ined but in a dif\/ferent manner
(see \cite[Sections~4.2.2 and~4.2.3]{EO07} for the def\/inition).
\end{Def}

Note that $F_g$ def\/ined here is dif\/ferent from $F_{g,n}$
def\/ined in the previous subsection.
$F_g$'s are also called {\it symplectic invariants}
since they are invariant under symplectic transformations
of the spectral curve (see~\cite{EO07}).
Explicit computation shows that
\begin{gather*}
F_{0}(t)   =   -\frac{48 q_{0}^{5}}{5},
\qquad
F_{1}(t)   =   -\frac{1}{24} \log(-3q_{0}), \\
F_{2}(t)   =   \frac{7}{207360 q_{0}^{5}},
\qquad
F_{3}(t)   =   \frac{245}{429981696 q_{0}^{10}}.
\end{gather*}

\begin{thm} [\protect{\cite{DGZJ} and \cite[Section~10.6]{EO07}}] \label{conj:1}
The generating function of the free energy $F_{g}(t)$
gives a $\tau$-function of~$\PI$:
\begin{gather*}
\log\tau(t,\hbar) = \sum_{g=0}^{\infty} \hbar^{2g-2} F_{g}(t).
\end{gather*}
Namely,
\begin{gather} \label{eq:tau-conjecture-g}
\frac{dF_{g}(t)}{dt} = \sigma_{2g}(t).
\end{gather}
\end{thm}

The proof will be given in Section~\ref{section:proof}.
It is worth mentioning that the closed free energies specify one
particular $\tau$-function although there is an ambiguity
in Def\/inition~\ref{def:tau-function}.

\begin{prop}
For $g \ge 2$, we have
\begin{gather} \label{eq:normalization-of-Fg}
F_{g}(t) = \int^{t}_{\infty} \sigma_{2g}(t')dt'.
\end{gather}
\end{prop}

\begin{proof}
Let us describe the behavior of the  $W_{g,n}$'s when
$q_{0} \rightarrow \infty$ (i.e., $t \rightarrow \infty$).
When $q_{0}$ tends to $\infty$,
no singular point of the integrand in the right hand-side of~\eqref{eq:top-rec}
on the $z$-plane hits the integration cycle $\gamma_{0}$.
Thus, we can show that
\begin{gather*}
W_{g,n}(z_{1},\dots,z_{n}) = O\big(q_{0}^{-(2g-2+n)}\big)
\end{gather*}
for $2g-2+n \ge 0$. This implies that
\begin{gather*}
F_{g}(t) = O\big(q_{0}^{-(2g-2)}\big)
\end{gather*}
holds since $\Phi(z) \sim q_{0}$ as $q_{0} \rightarrow \infty$
(but we can verify that~$F_{g}$ for $g \ge 2$
has a stronger decay in the above explicit computations).
This completes the proof of~\eqref{eq:normalization-of-Fg}.
\end{proof}

\subsection{Asymptotics of Eynard--Orantin dif\/fernetials}
\label{section:asymptotics-Wgn}

The rest of this section will be devoted to
show some important properties of $W_{g,n}$ and $F_{g,n}$.
Firstly, we will describe the asymptotic behavior
of them near $z_{i} = \infty$.

\begin{lem} \label{lem:decay-Wgn-Fgn}\quad
\begin{itemize}\itemsep=0pt
\item[$(i)$]
For $2g-2+n \ge 0$, we have
\begin{gather} \label{eq:asympt-Wgn}
W_{g,n}(z_{1},\dots,z_{n}) =
\left(\frac{c_{g,n}}{z_{1}^{2}\cdots z_{n}^{2}}
+ O\big(z_{1}^{-4}\cdots z_{n}^{-4}\big) \right)dz_{1}\cdots dz_{n}
\end{gather}
as $z_{i} \rightarrow \infty$ for all $i=1,\dots,n$.
Here $c_{g,n} \in \bbC$ is a constant.
\item[$(ii)$] %
For $2g-2+n \ge 0$, we have
\begin{gather} \label{eq:asympt-Fgn}
F_{g,n}(z_{1},\dots,z_{n}) =
\frac{c'_{g,n}}{z_{1}\cdots z_{n}} +
O\big(z_{1}^{-3}\cdots z_{n}^{-3}\big), \qquad
c'_{g,n} \in \bbC,
\end{gather}
as $z_{i} \rightarrow \infty$ for all $i=1,\dots,n$.
\end{itemize}
\end{lem}

\begin{proof}
The f\/irst property \eqref{eq:asympt-Wgn} follows from the
analyticity of $W_{g,n}$ at $z_i = \infty$.
The second property~\eqref{eq:asympt-Fgn} follows
from~\eqref{eq:asympt-Wgn} immediately
because $F_{g,n}(z_{1},\dots,z_{n})$ doesn't have
a constant term due to the def\/inition \eqref{eq:open-Fgn}.
\end{proof}

As a corollary, the principal specialization
of open free energies satisf\/ies
\begin{gather} \label{eq:decay-of-Fgn-prin}
F_{g,n}(z,\dots,z) = O\big(z^{-n}\big)
\end{gather}
when $z \rightarrow \infty$.

\subsection{Variation of spectral curve}
\label{section:variation-formula}

There is a formula (for ``variation of spectral curves'')
that allows us to compute derivatives of~$W_{g,n}$ etc.\
with respect to the parameter~$t$.

\begin{thm} [cf.~\protect{\cite[Theorem~5.1]{EO07}}]\label{thm:variation}\quad
\begin{itemize}\itemsep=0pt
\item[$(i)$]
For $2g-2+n \ge 0$, we have
\begin{gather}
\frac{\p}{\p t} W_{g,n}(z(x_{1}), \dots, z(x_{n}))\nonumber\\
\qquad{}
= -2\Res_{x_{n+1}=\infty} z(x_{n+1})
W_{g,n+1}\bigl( z(x_{1}),\dots,z(x_{n}),z(x_{n+1}) \bigr).\label{eq:variation-Wgn}
\end{gather}
\item[$(ii)$]
For $g \ge 1$, we have
\begin{gather} \label{eq:dt-closed-Fg}
\frac{dF_{g}}{dt}(t)
 = - 2\Res_{x=\infty} z(x) W_{g,1}(z(x))
= -\Res_{z=\infty} z W_{g,1}(z).
\end{gather}
\item[$(iii)$]
For $2g-2+ n \ge 1$, we have
\begin{gather*} 
\frac{\p}{\p t} F_{g,n}(z(x_{1}),\dots,z(x_{n})) \\
\qquad{} =
- 2\Res_{x_{n+1}=\infty} z(x_{n+1})
d_{x_{n+1}} F_{g,n+1}(z(x_{1}),\dots,z(x_{n}),z(x_{n+1})),
\end{gather*}
or equivalently,
\begin{gather}
\frac{\p}{\p t} F_{g,n}(z(x_{1}),\dots,z(x_{n})) \nonumber\\
\qquad{}= \lim_{z_{n+1} \rightarrow \infty} \left.
\left( z_{n+1}^{2}
\frac{\p}{\p z_{n+1}}F_{g,n+1}(z_{1},\dots,z_{n},z_{n+1})\right)
\right|_{(z_{1},\dots,z_{n})=(z(x_{1}), \dots, z(x_{n}))}. \label{eq:variation-openFgn-2}
\end{gather}
\end{itemize}
\end{thm}

\begin{proof}
Set $\Lambda(z) := z$.
Then, we can check $\Lambda(z)$ satisf\/ies
the required condition
\begin{gather*} 
\Res_{z=\infty} \left( \Lambda(z) W_{0,2}(z,z_{1}) \right) = - dz_{1}
= - \left(\frac{\p y}{\p t}(z_{1}) dx(z_{1}) -
\frac{\p x}{\p t}(z_{1}) dy(z_{1}) \right)
\end{gather*} 
to apply \cite[Theorem~5.1]{EO07}. Thus the claim~(i) and~(ii)
are proved. Integrating both hand-sides of
\eqref{eq:variation-Wgn}, we have~(iii).
\end{proof}

\subsection{Dif\/ferential recursion for open free energies}

Here we give a key theorem in the proof of our main results.
We have the following {\em differential recursion} which is
a modif\/ication of the one obtained in \cite{DM14, DM14-2}.

\begin{thm} \label{thm:diff-rec}
The open free energies for $2g-2+n \ge 2$ satisfy
the following equations
\begin{gather}
\frac{\p F_{g,n}}{\p z_{1}}(z_{1},\dots,z_{n}) =
\sum_{j=2}^{n} \frac{-2z_{j}}{z_{1}^{2}-z_{j}^{2}}
\biggl( \frac{1}{2y(z_{1})\frac{dx}{dz}(z_{1})}
\frac{\p F_{g,n-1}}{\p z_{1}}(z_{[\hat{j}]})
 - \frac{1}{2y(z_{j})\frac{dx}{dz}(z_{j})}
\frac{\p F_{g,n-1}}{\p z_{j}}(z_{[\hat{1}]}) \biggr)
\nonumber\\
\hphantom{\frac{\p F_{g,n}}{\p z_{1}}(z_{1},\dots,z_{n}) =}{}
- \frac{1}{2y(z_{1})\frac{dx}{dz}(z_{1})} \frac{\p^{2}}{\p u_{1} \p u_{2}}
\biggl( F_{g-1,n+1}(u_{1},u_{2},z_{[\hat{1}]}) \nonumber\\
\hphantom{\frac{\p F_{g,n}}{\p z_{1}}(z_{1},\dots,z_{n}) =}{}
+
\sum_{\substack{g_{1}+g_{2}=g \\ I \sqcup J = [\hat{1}]}}^{\rm stable}
F_{g_{1}, |I|+1}(u_{1},z_{I}) F_{g_{2}, |J|+1}(u_{2}, z_{J})
\biggr)\Biggr|_{u_{1}=u_{2}=z_{1}}
\nonumber\\
\hphantom{\frac{\p F_{g,n}}{\p z_{1}}(z_{1},\dots,z_{n}) =}{}
+ \frac{s}{\frac{dy}{dz}(s) \frac{dx}{dz}(s)(z_{1}^{2}-s^{2})}
\Biggl[ \sum_{j=2}^{n} \frac{-2 z_{j}}{z_{j}^{2}-s^{2}}
\frac{\p F_{g,n-1}}{\p z_{1}}(s,z_{[\hat{1}, \hat{j}]})
\nonumber\\
\hphantom{\frac{\p F_{g,n}}{\p z_{1}}(z_{1},\dots,z_{n}) =}{}
+ \frac{\p^{2}}{\p u_{1} \p u_{2}}
\biggl( F_{g-1,n+1}(u_{1},u_{2},z_{[\hat{1}]})\nonumber\\
\hphantom{\frac{\p F_{g,n}}{\p z_{1}}(z_{1},\dots,z_{n}) =}{}
 +
\sum_{\substack{g_{1}+g_{2}=g \\
I \sqcup J = [\hat{1}]}}^{\rm stable}
F_{g_{1}, |I|+1}(u_{1},z_{I}) F_{g_{2}, |J|+1}(u_{2}, z_{J})
\biggr)\Biggr|_{u_{1}=u_{2}=s} \Biggr].\label{eq:diff-rec}
\end{gather}
Here
$s = (3q_{0})^{1/2}$
is a zero of~$y(z)$.
\end{thm}

\begin{proof}
This can be proved by a similar technique used in
\cite[Theorem~4.7]{DM14}, as follows. Integrating
the topological recursion relation~\eqref{eq:top-rec}
with respect to $z_{2}, \dots, z_{n}$, we have
\begin{gather}
\frac{\p}{\p z_{1}} F_{g,n}(z_{1},\dots,z_{n})  =
\frac{1}{2^{n-1}}\int_{\bar{z}_{2}}^{z_{2}} \cdots
\int_{\bar{z}_{n}}^{z_{n}} W_{g,n}(z_{1},\dots,z_{n})\nonumber \\
\hphantom{\frac{\p}{\p z_{1}} F_{g,n}(z_{1},\dots,z_{n})}{}
=
\frac{1}{2\pi i} \frac{1}{2^{n-1}} \oint_{\gamma_{0}}
K(z,z_{1}) R_{g,n}(z,z_{2},\dots,z_{n}), \label{eq:pre-diff-recursion}
\end{gather}
where
\begin{gather*}
R_{g,n}(z,z_{2},\dots,z_{n}) =
\sum_{j=2}^{n}
\Biggl[ \biggl(\int^{z_{j}}_{\bar{z}_{j}}W_{0,2}(z,z_{j}) \biggr)
\biggl(\int^{z_{[\hat{1}, \hat{j}]}}_{\bar{z}_{[\hat{1}, \hat{j}]}}
W_{g,n-1}(\bar{z}, z_{[\hat{1},\hat{j}]}) \biggr)\\
\hphantom{R_{g,n}(z,z_{2},\dots,z_{n}) =}{}
- \biggl(\int^{z_{j}}_{\bar{z}_{j}}W_{0,2}(\bar{z},z_{j}) \biggr)
\biggl(\int^{z_{[\hat{1}, \hat{j}]}}_{\bar{z}_{[\hat{1}, \hat{j}]}}
W_{g,n-1}(z, z_{[\hat{1},\hat{j}]})
\biggr) \Biggr]
\\
\hphantom{R_{g,n}(z,z_{2},\dots,z_{n}) =}{}
+ \int^{z_{[\hat{1}]}}_{\bar{z}_{[\hat{1}]}}
W_{g-1,n+1}(z,\bar{z},z_{[\hat{1}]})\\
\hphantom{R_{g,n}(z,z_{2},\dots,z_{n}) =}{}
 + \sum_{\substack{g_{1}+g_{2}=g \\
 I \sqcup J = [\hat{1}]}}^{\text{stable}}
\biggl( \int^{z_{I}}_{\bar{z}_{I}}
W_{g_{1}, |I|+1}(z,z_{I}) \biggr)
\biggl( \int^{z_{J}}_{\bar{z}_{J}}
W_{g_{2}, |J|+1}(\bar{z}, z_{J})
\biggr).
\end{gather*}
Here, for a set $L =\{\ell_{1}, \dots, \ell_{k} \}
\subset \{1,\dots,n \}$ of indices, we have used the notation
\begin{gather*}
\int^{z_{L}}_{\bar{z}_{L}} W_{g,n}(z_{1},\dots,z_{n}) :=
\int^{z_{\ell_{1}}}_{\bar{z}_{\ell_{1}}} \cdots
\int^{z_{\ell_{k}}}_{\bar{z}_{\ell_{k}}} W_{g,n}(z_{1},\dots,z_{n}).
\end{gather*}
On the $z$-plane, the integrand
$K(z,z_{1})R_{g,n}(z,z_{1},\dots,z_{n})$
in the right hand-side of~\eqref{eq:pre-diff-recursion}
has poles at
\begin{itemize}\itemsep=0pt
\item
at $z = z_{1}, \bar{z}_{1}$ which are poles of $K(z,z_{1})$,
\item
at $z = z_{2}, \dots, z_{n}, \bar{z}_{2}, \dots, \bar{z}_{n}$
which are poles of $W_{0,2}(z,z_{j})$ and $W_{0,2}(\bar{z},z_{j})$,
\item
at $z = s, \bar{s}$ which are poles of $K(z,z_{1})$,
\end{itemize}
and all of them are simple poles.
Then, the equalities
\begin{gather*}
\int^{z_{j}}_{\bar{z}_{j}} W_{0,2}(z,z_{j})   =
\left(\frac{1}{z-z_{j}} - \frac{1}{z-\bar{z}_{j}} \right)dz, \\
\frac{1}{2^{n-2}}
\int^{z_{[\hat{1}, \hat{j}]}}_{\bar{z}_{[\hat{1}, \hat{j}]}}
W_{g,n}(z,z_{[\hat{1}, \hat{j}]})   =
\frac{\p F_{g,n-1}}{\p z_{1}}(z,z_{[\hat{1}, \hat{j}]})
\end{gather*}
and the residue theorem show~\eqref{eq:diff-rec}.
\end{proof}

\begin{rem} \label{rem:difference-from-DM}
Note that the f\/irst two blocks in the right hand-side
of~\eqref{eq:diff-rec} coincide with that obtained
in~\cite{DM14, DM14-2}.
Unlike the case of \cite{DM14, DM14-2},
we need more terms arising from $z=s$ corresponding
to the singular point $(x,y)=(q_{0},0)$
of the spectral curve~\eqref{eq:sp-curve}
since it becomes a~(simple) pole of
the recursion kernel $K(z,z_{1})$.
It also worth mentioning that the
right hand-side of~\eqref{eq:diff-rec} doesn't have
singularity at $z_{j} = s$ for $j=1,\dots,n$.
\end{rem}

Using this dif\/ferential recursion, we can give
an alternative expression of~\eqref{eq:variation-openFgn-2}
as follows.

\begin{thm} \label{thm:variation-of-open-free-energy}
For $2g-2+n \ge 1$, the following holds:
\begin{gather} \label{eq:t-derivative-recursion}
\frac{\p }{\p t} F_{g,n}(z(x_{1}),\dots,z(x_{n}))
= E_{g,n}(z(x_{1}),\dots,z(x_{n})),
\end{gather}
where
\begin{gather}
E_{g,n}(z_{1},\dots,z_{n}) \nonumber\\
\qquad{}
:=\sum_{j=1}^{n}
\frac{2z_{j}}{2y(z_{j})\frac{dx}{dz}(z_{j})}
\frac{\p F_{g,n}}{\p z_{j}}(z_{1},\dots,z_{n}) +
\frac{s}{\frac{dy}{dz}(s) \frac{dx}{dz}(s)}
\sum_{j=1}^{n} \frac{-2 z_{j}}{z_{j}^{2}-s^{2}}
\frac{\p F_{g,n}}{\p u_{1}}(u_{1},z_{[\hat{j}]})
\Biggr|_{u_{1}=s}
\nonumber\\
\qquad\quad{}+ \frac{s}{\frac{dy}{dz}(s) \frac{dx}{dz}(s)}
\frac{\p^{2}}{\p u_{1} \p u_{2}}
\biggl( F_{g-1,n+2}(u_{1},u_{2},z_{1},\dots,z_{n})
\nonumber\\
\qquad\quad{}+ \sum_{\substack{g_{1}+g_{2}=g \\
I \sqcup J = \{1,\dots,n\}}}^{\rm stable}
F_{g_{1}, |I|+1}(u_{1},z_{I}) F_{g_{2}, |J|+1}(u_{2}, z_{J})
\biggr)\Biggr|_{u_{1}=u_{2}=s}.\label{eq:Egn}
\end{gather}
\end{thm}

\begin{proof}
The equality \eqref{eq:variation-openFgn-2} shows that
the left hand-side of~\eqref{eq:t-derivative-recursion}
coincides with
\begin{gather*}
\lim_{z_{n+1} \rightarrow \infty} z_{n+1}^{2}
\frac{\p}{\p z_{n+1}} F_{g,n+1}(z_{1},\dots,z_{n},z_{n+1})
\end{gather*}
after the substitution $z_{i} \mapsto z(x_{i})$ for $i=1,\dots,n$.
Then, the equality follows from
the asymptotic behavior \eqref{eq:asympt-Fgn} of $F_{g,n}$'s
and the above dif\/ferential recursion \eqref{eq:diff-rec}
for $2g-2+ (n+1) \ge 2$.
\end{proof}

\section{Proof of main theorems}
\label{section:proof}

\subsection{Strategy for the proof}

What we will show here is that the formal series $S(x,t,\hbar)$
def\/ined in \eqref{eq:Riccati-solution} satisf\/ies
the system of equations \eqref{eq:Riccati-2} and
\eqref{eq:dt-Riccati-2}. In addition, we will also prove
the equality \eqref{eq:tau-conjecture-g}.
These equalities will be proved by an induction as follows.

\begin{thm} \label{thm:induction-scheme}
Let $[\bullet]_{\hbar^{m}}$ be the coefficient of
$\hbar^{m}$ in a formal series $\bullet$ of $\hbar$.
For an even integer $k \ge 2$, assume that
\begin{gather}
\frac{\p S_{m}}{\p x}(x,t) = P_{m}(x,t)\qquad
\text{for} \quad m=0,\dots,k-1, \nonumber\\
\frac{\p S_{m}}{\p t}(x,t) = \left[ \frac{1}{2(x-q)}\left(
\hbar \frac{\p S}{\p x} - p \right) \right]_{\hbar^{m}} \qquad
\text{for} \quad m=0,\dots,k-1, \nonumber\\
\frac{d F_{g}}{d t}(t) = \sigma_{2g}(t)  \qquad
\text{for} \quad g = k/2
\label{eq:assume-A-1}
\end{gather}
holds. Here $P_{m}(x,t) = P^{(+)}_{m}(x,t)$
is the coefficient of $\hbar^{m-1}$ in the formal solution
$P^{(+)}(x,t,\hbar)$ of the Riccati equation~\eqref{eq:Riccati}
constructed in Section~{\rm \ref{subsection:WKB-for-scalar-Lax}},
and~$\sigma_{2g}$ is given in~\eqref{eq:parity-sigma}.
Then, we have
\begin{itemize}\itemsep=0pt
\item[$(A)$] %
The following equality holds for $m = k$ and $k+1$:
\begin{gather}
\label{eq:claim-A-1}
\left[
\hbar^{2}\left( \left(\frac{\p S}{\p x}\right)^{2} +
\frac{\p^{2} S}{\p x^{2}} \right)
\right]_{\hbar^{m}}
 =  \left[ 2\hbar^{2} \frac{\p S}{\p t} +
\big(4x^{3}+2t x + p^{2} - 4q^{3}-2tq\big) \right]_{\hbar^{m}}.
\end{gather}

\item[$(B)$]
The following equalities hold:
\begin{gather}
\label{eq:claim-A-3}
\frac{\p S_{k}}{\p x}(x,t)   =   P_{k}(x,t), \qquad
\frac{\p S_{k}}{\p t}(x,t) = \left[ \frac{1}{2(x-q)}\left(
\hbar \frac{\p S}{\p x} - p \right) \right]_{\hbar^{k}}, \\
\frac{\p S_{k+1}}{\p x}(x,t)   =   P_{k+1}(x,t), \qquad
\frac{\p S_{k+1}}{\p t}(x,t) = \left[ \frac{1}{2(x-q)}\left(
\hbar \frac{\p S}{\p x} - p \right) \right]_{\hbar^{k+1}}, \label{eq:claim-B-3}
\\
\label{eq:claim-B-4}
\frac{dF_{g}}{dt}(t)   =    \sigma_{2g}(t)
\qquad \text{for} \quad g=(k+2)/2.
\end{gather}
\end{itemize}
\end{thm}

It is obvious that our main theorems (Theorems~\ref{conj:2} and~\ref{conj:1})
follow from the statements
in~$(A)$ and~$(B)$. The rest of this section is devoted to
give a proof of~$(A)$ and~$(B)$.

\subsection[Proof of $(A)$]{Proof of $\boldsymbol{(A)}$}
\label{section:proof-of-A}

We emphasize that the results shown in Section~\ref{subsec:proof-A-1} below
are proved without using the assumption~\eqref{eq:assume-A-1}.
We also note that we only use the second equality
in assumption~\eqref{eq:assume-A-1} in Section~\ref{subsec:proof-A-2}
to prove~$(A)$.

\subsubsection{Computation of principal specializations}
\label{subsec:proof-A-1}

Def\/ine
\begin{gather}
G_{g,n}(z_{1},\dots,z_{n}) :=
\frac{\p F_{g,n}}{\p z_{1}}(z_{1},\dots,z_{n})-\sum_{j=2}^{n} \frac{-2z_{j}}{z_{1}^{2}-z_{j}^{2}}
\biggl( \frac{1}{2y(z_{1})\frac{dx}{dz}(z_{1})}
\frac{\p F_{g,n-1}}{\p z_{1}}(z_{[\hat{j}]})
\nonumber\\
\hphantom{G_{g,n}(z_{1},\dots,z_{n}) :=}{}
 - \frac{1}{2y(z_{j})\frac{dx}{dz}(z_{j})}
\frac{\p F_{g,n-1}}{\p z_{j}}(z_{[\hat{1}]}) \biggr)
\nonumber\\
\hphantom{G_{g,n}(z_{1},\dots,z_{n}) :=}{}
+ \frac{1}{2y(z_{1})\frac{dx}{dz}(z_{1})} \frac{\p^{2}}{\p u_{1} \p u_{2}}
\biggl( F_{g-1,n+1}(u_{1},u_{2},z_{[\hat{1}]})\nonumber\\
\hphantom{G_{g,n}(z_{1},\dots,z_{n}) :=}{}
+ \sum_{\substack{g_{1}+g_{2}=g \\ I \sqcup J = [\hat{1}]}}^{\rm stable}
F_{g_{1}, |I|+1}(u_{1},z_{I}) F_{g_{2}, |J|+1}(u_{2}, z_{J})
\biggr)\Biggr|_{u_{1}=u_{2}=z_{1}}.
\label{eq:Ggn}
\end{gather}

The technique developed in
\cite{DM14, DM14-2} enables us to show the following.

\begin{lem} [cf.~\protect{\cite[Theorem 6.5]{DM14}}]
\label{prop:partial-Riccati}
For $m \ge 2$, we have
\begin{gather} \label{eq:partial-Riccati}
\Biggl(\frac{2y(z)}{\frac{dx}{dz}(z)}
\sum_{\substack{2g-2+n=m \\ g\ge0, n\ge 1}}
\frac{G_{g,n}(z,\dots,z)}{(n-1)!} \Biggr)\Biggl|_{z=z(x)} =
\sum_{\substack{a+b=m+1 \\ a,b \ge 0}}
\frac{\p S_{a}}{\p x} \frac{\p S_{b}}{\p x} +
\frac{\p^{2} S_{m}}{\p x^{2}} - \frac{1}{x-q_{0}} \frac{\p S_{m}}{\p x}.
\end{gather}
\end{lem}

\begin{proof}
As is shown in \cite[Theorem 6.5]{DM14},
applying $\sum\limits_{2g-2+n=m} \frac{1}{(n-1)!}$
and the principal specia\-li\-zation to~\eqref{eq:Ggn}, we have
\begin{gather*}
\sum_{\substack{2g-2+n=m \\ g\ge0, n\ge 1}}
\frac{G_{g,n}(z,\dots,z)}{(n-1)!} = \frac{1}{2y(z)\frac{dx}{dz}(z)}
\Biggl( \sum_{\substack{a+b=m+1 \\ a,b \ge 2}}
\frac{\p S_{a}(x(z))}{\p z} \frac{\p S_{b}(x(z))}{\p z} +
\frac{\p^{2}S_{m}(x(z))}{\p z^{2}} \Biggr)
\\
\hphantom{\sum_{\substack{2g-2+n=m \\ g\ge0, n\ge 1}}
\frac{G_{g,n}(z,\dots,z)}{(n-1)!} =}{}
+ \frac{\p S_{m+1}(x(z))}{\p z}
+ \left\{ \frac{\p }{\p z}\left(
\frac{1}{2 y(z) \frac{dx}{dz}(z)} \right) \right\}
\frac{\p S_{m}(x(z))}{\p z}.
\end{gather*}
After the coordinate change $z = z(x)$, the right hand-side becomes
\begin{gather*}
\frac{\frac{dx}{dz}(z(x))}{2y(z(x))}\Biggl(
\sum_{\substack{a+b=m+1 \\ a,b \ge 2}}
\frac{\p S_{a}}{\p x} \frac{\p S_{b}}{\p x} +
\frac{\p^{2}S_{m}}{\p x^{2}} + 2y(z(x)) \frac{\p S_{m+1}}{\p x}
- \frac{1}{y(z(x))}\frac{\p y(z(x))}{\p x} \frac{\p S_{m}}{\p x} \Biggr).
\end{gather*}
Then, the desired equality \eqref{eq:partial-Riccati} follows
from the above equality and
\begin{gather*}
\frac{\p S_{0}}{\p x} = y(z(x)), \qquad
\frac{\p S_{1}}{\p x} =
-\frac{1}{2 y(z(x))} \frac{\p y(z(x))}{\p x} + \frac{1}{2(x-q_{0})}.\tag*{\qed}
\end{gather*}
\renewcommand{\qed}{}
\end{proof}

Note that the right hand-side of~\eqref{eq:partial-Riccati} coincides with
\begin{gather*}
\biggl[ \hbar^{2}\biggl( \left(\frac{\p S}{\p x}\right)^{2} +
\frac{\p^{2} S}{\p x^{2}} \biggr)  \biggr]_{\hbar^{m+1}}
- \frac{1}{x-q_{0}} \frac{\p S_{m}}{\p x}.
\end{gather*}
Thus, Lemma \ref{prop:partial-Riccati} relates
the principal specialization of~$G_{g,n}$
to the left hand-side of~\eqref{eq:claim-A-1}.
Next we also relate them to the right hand-side of~\eqref{eq:claim-A-1}.

\begin{lem} \label{lemma:sum-prin-G-and-E}
Let $E_{g,n}(z_{1},\dots,z_{n})$
be the functions defined by~\eqref{eq:Egn}.
Then, the following equality holds for $m \ge 2$
\begin{gather} \label{eq:sum-prin-G-and-E}
\sum_{\substack{2g-2+n=m \\ g \ge 0, n\ge 2 }}
\left(
\frac{2y(z)}{\frac{dx}{dz}(z)}
\frac{G_{g,n}(z,\dots,z)}{(n-1)!} -
\frac{2 E_{g,n-1}(z,\dots,z)}{(n-1)!}\right)\Biggr|_{z=z(x)}
 = -\frac{1}{x-q_{0}} \frac{\p S_{m}}{\p x}.
\end{gather}
\end{lem}

\begin{proof}
Theorem \ref{thm:diff-rec} shows that
\eqref{eq:Ggn} can also be written as
\begin{gather}
G_{g,n}(z_{1},\dots,z_{n}) =
\frac{s}{\frac{dy}{dz}(s) \frac{dx}{dz}(s)(z_{1}^{2}-s^{2})}
\sum_{j=2}^{n} \frac{-2 z_{j}}{z_{j}^{2}-s^{2}}
\frac{\p F_{g,n-1}}{\p z_{1}}(s,z_{[\hat{1}, \hat{j}]})\nonumber
\\
\hphantom{G_{g,n}(z_{1},\dots,z_{n}) =}{}
+ \frac{s}{\frac{dy}{dz}(s) \frac{dx}{dz}(s)(z_{1}^{2}-s^{2})}
\frac{\p^{2}}{\p u_{1} \p u_{2}}
\biggl( F_{g-1,n+1}(u_{1},u_{2},z_{[\hat{1}]}) \nonumber\\
\hphantom{G_{g,n}(z_{1},\dots,z_{n}) =}{}+
\sum_{\substack{g_{1}+g_{2}=g \\ I \sqcup J = [\hat{1}]}}^{\rm stable}
F_{g_{1}, |I|+1}(u_{1},z_{I}) F_{g_{2}, |J|+1}(u_{2}, z_{J})
\biggr)\Biggr|_{u_{1}=u_{2}=s}.\label{eq:Ggn-2}
\end{gather}
Taking the principal specialization of \eqref{eq:Ggn-2} and~\eqref{eq:Egn} with $n \mapsto n-1$, we have
\begin{gather}
\frac{2y(z)}{\frac{dx}{dz}(z)}
\frac{G_{g,n}(z,\dots,z)}{(n-1)!} -
\frac{2 E_{g,n-1}(z,\dots,z)}{(n-1)!} \nonumber\\
\qquad{} =
- \frac{4z}{2y(z)\frac{dx}{dz}(z)}
\left(\frac{1}{(n-1)!} \frac{\p}{\p z} F_{g,n-1}(z,\dots,z)\right)\label{eq:prin-G-and-E}
\end{gather}
for any $g \ge 0$ and $n \ge 2$ satisfying
$2g-2+n \ge 2$.
Then, summing up~\eqref{eq:prin-G-and-E}
for $g \ge 0$, $n \ge 2$ satisfying $2g-2+n = m$,
we obtain~\eqref{eq:sum-prin-G-and-E} after
the coordinate change $z = z(x)$.
\end{proof}

On the other hand, Theorem~\ref{thm:variation-of-open-free-energy} implies that
\begin{gather*} 
\sum_{\substack{2g-2+n=m \\ g \ge 0,\,  n\ge 2 }}
\frac{2E_{g,n-1}(z,\dots,z)}{(n-1)!} \Biggr|_{z=z(x)}
= 2\frac{\p}{\p t} S_{m}
= \left[ 2\hbar^{2} \frac{\p S}{\p t} \right]_{\hbar^{m+1}}
\end{gather*}
holds for $m \ge 2$.
Therefore we have the following.

\begin{lem} \label{lem:tderivative-and-partial-Riccati}
The equality
\begin{gather*}
\left[ \hbar^{2}\left( \left(\frac{\p S}{\p x}\right)^{2} +
\frac{\p^{2} S}{\p x^{2}} \right)  \right]_{\hbar^{m+1}} \\
\qquad{} =
\left[ 2\hbar^{2} \frac{\p S}{\p t} \right]_{\hbar^{m+1}} + \sum_{\substack{2g-2+n=m \\ g \ge 0,\, n\ge {1} }}
\frac{2y(z)}{\frac{dx}{dz}(z)}
\frac{G_{g,n}(z,\dots,z)}{(n-1)!} -
\sum_{\substack{2g-2+n=m \\ g \ge 0, n\ge {2} }}
\frac{2y(z)}{\frac{dx}{dz}(z)}
\frac{G_{g,n}(z,\dots,z)}{(n-1)!}
\end{gather*}
holds for $m \ge 2$.
\end{lem}

\subsubsection[Completion of the proof of $(A)$]{Completion of the proof of $\boldsymbol{(A)}$}
\label{subsec:proof-A-2}

Lemma \ref{lem:tderivative-and-partial-Riccati} implies
\begin{gather}
\left[ \hbar^{2}\left( \left(\frac{\p S}{\p x}\right)^{2}
+ \frac{\p^{2} S}{\p x^{2}} \right)  \right]_{\hbar^{m+1}}
= \left[ 2\hbar^{2} \frac{\p S}{\p t} \right]_{\hbar^{m+1}}
 \qquad \text{if $m$ is even},\nonumber\\
\left[ \hbar^{2}\left( \left(\frac{\p S}{\p x}\right)^{2}
+ \frac{\p^{2} S}{\p x^{2}} \right)  \right]_{\hbar^{m+1}}
= \left[ 2\hbar^{2} \frac{\p S}{\p t} \right]_{\hbar^{m+1}}
+ \frac{2y}{\frac{dx}{dz}} G_{(m+1)/2,1}
\qquad \text{if $m$ is odd.}
\label{eq:compare-coefficients}
\end{gather}
On the other hand, it follows from \eqref{eq:parity-sigma} that
\begin{gather*}
\bigl[4x^{3}+2tx + p^{2} -4q^{3} -2 tq \bigr]_{\hbar^{m+1}}
=
\begin{cases}
0 & \text{if $m$ is even}, \\
2\sigma_{m+1} & \text{if $m$ is odd}.
\end{cases}
\end{gather*}
Therefore, under the assumption \eqref{eq:assume-A-1},
the desired equality~\eqref{eq:claim-A-1}
follows from~\eqref{eq:compare-coefficients}
and Lem\-ma~\ref{lem:Gg1} below.

\begin{lem} \label{lem:Gg1}
For $g \ge 2$, we have
\begin{gather} \label{eq:n-is-1-part-in-Ggn}
\frac{2y(z)}{\frac{dx}{dz}(z)} G_{g,1}(z) =
 2 \frac{dF_{g}}{dt}(t).
\end{gather}
\end{lem}

\begin{proof}
Firstly, we note that
\begin{gather} \label{eq:explicit-Gg1}
\frac{2y(z)}{\frac{dx}{dz}(z)} G_{g,1}(z) =
\frac{1}{4s^{2}}
\biggl(
\frac{\p^{2} F_{g-1,2}}{\p z_{1} \p z_{2}}(s,s) +
\sum_{\substack{g_{1}+g_{2} = g \\ g_{1}, g_{2} \ge 1}}
\frac{\p F_{g_{1}, 1}}{\p z_{1}}(s)
\frac{\p F_{g_{2}, 1}}{\p z_{2}}(s)
\biggr)
\end{gather}
holds. Using the dif\/fernetial recursion \eqref{eq:diff-rec}
for $n=1$, we have
\begin{gather*}
\frac{\p F_{g,1}}{\p z_{1}} (z) =
- \frac{1}{2y(z)\frac{dx}{dz}(z)}
\biggl(
\frac{\p^{2} F_{g-1,2}}{\p z_{1} \p z_{2}}(z,z) +
\sum_{\substack{g_{1}+g_{2} = g \\ g_{1}, g_{2} \ge 1}}
\frac{\p F_{g_{1}, 1}}{\p z_{1}}(z)
\frac{\p F_{g_{2}, 1}}{\p z_{1}}(z)
\biggr)
\\
\hphantom{\frac{\p F_{g,1}}{\p z_{1}} (z) =}{}
+ \frac{s}{\frac{dy}{dz}(s)\frac{dx}{dz}(s)(z^{2}-s^{2})}
\biggl(
\frac{\p^{2} F_{g-1,2}}{\p z_{1} \p z_{2}}(s,s) +
\sum_{\substack{g_{1}+g_{2} = g \\ g_{1}, g_{2} \ge 1}}
\frac{\p F_{g_{1}, 1}}{\p z_{1}}(s)
\frac{\p F_{g_{2}, 1}}{\p z_{1}}(s)
\biggr).
\end{gather*}
Then, Lemma~\ref{lem:decay-Wgn-Fgn}
implies that
\begin{gather*}
z\frac{\p F_{g,1}}{\p z_{1}} (z)dz
= z W_{g,1}(z)  = \frac{1}{8s^{2}} \biggl(
\frac{\p^{2} F_{g-1,2}}{\p z_{1} \p z_{2}}(s,s) +
\sum_{\substack{g_{1}+g_{2} = g \\ g_{1}, g_{2} \ge 1}}
\frac{\p F_{g_{1}, 1}}{\p z_{1}}(s)
\frac{\p F_{g_{2}, 1}}{\p z_{2}}(s)
\biggr) \frac{dz}{z} + O(1)
\end{gather*}
holds when $z \rightarrow \infty$.
Then the equality~\eqref{eq:n-is-1-part-in-Ggn}
follows from~\eqref{eq:dt-closed-Fg} and~\eqref{eq:explicit-Gg1}.
\end{proof}

\subsection[Proof of $(B)$]{Proof of $\boldsymbol{(B)}$}

One of the desired equality \eqref{eq:claim-A-3} is proved as follows.

\begin{lem} \label{prop:integral-for-Sm}
Under the assumption \eqref{eq:assume-A-1}, we have
\begin{gather}
\label{eq:claim-A-5}
\frac{\p S_{k}}{\p x}(x,t)  =   P_{k}(x,t), \\
\label{eq:claim-A-4}
S_{k}(x,t)   =   \int^{x}_{\infty} P_{k}(x',t)dx', \\
\label{eq:claim-A-2}
\frac{\p S_{k}}{\p t}(x,t)   =   \left[\frac{1}{2(x-q)}
\left(\hbar \frac{\p S}{\p x} - p \right) \right]_{\hbar^{k}}.
\end{gather}
\end{lem}

\begin{proof}
The equality \eqref{eq:claim-A-1} for $m=k$ and
the second equality in the assumption~\eqref{eq:assume-A-1}
imply
\begin{gather*}
\left[
\hbar^{2}\left( \left(\frac{\p S}{\p x}\right)^{2} +
\frac{\p^{2} S}{\p x^{2}} \right)
\right]_{\hbar^{k}}
= \left[ \frac{\hbar}{x-q}\left( \hbar \frac{\p S}{\p x} - p \right)
+ \big(4x^{3}+2t x + p^{2} - 4q^{3}-2tq\big) \right]_{\hbar^{k}}.
\end{gather*}
Thus ${\p S_{k}}/{\p x}$ and $P_{k}$ satisfy the
same equation~\eqref{eq:recursion-for-Pm} under our induction hypothesis.
Then the uniqueness of the solution of~\eqref{eq:recursion-for-Pm}
implies~\eqref{eq:claim-A-5}.

Since $S_{m}(x)$ for $m \ge 2$ decay when
$x \rightarrow \infty$
(cf.~\eqref{eq:decay-of-Fgn-prin}),
the equality~\eqref{eq:claim-A-4}
immediately follows from~\eqref{eq:claim-A-5}.
Then, the equality~\eqref{eq:dt-P} shows
\begin{gather*}
\frac{\p}{\p t} S_{k}(x,t)   =
\int^{x}_{\infty} \left[ \hbar \frac{\p}{\p t} P
\right]_{\hbar^{k}}dx
  =
\int^{x}_{\infty} \frac{\p}{\p x} \left[
\frac{\hbar P - p}{2(x-q)} \right]_{\hbar^{k}} dx
  =
\left[
\frac{1}{2(x-q)} \left( \hbar \frac{\p S}{\p x} - p \right)
\right]_{\hbar^{k}}.
\end{gather*}
The last equality follows from the assumption~\eqref{eq:assume-A-1} and the fact that
$P_{m}(x,t)$'s decay when $x \rightarrow \infty$
for $m \ge 1$ (see Remark~\ref{rem:asymptotics}),
and
\begin{gather*}
\lim_{x\rightarrow\infty} \frac{P_{0}(x,t)}{(x-q_{0})^{2}} = 0.
\end{gather*}
Thus we have proved \eqref{eq:claim-A-2}.
\end{proof}

Since we have also already proved \eqref{eq:claim-A-1}
for $m = k+1$, we can prove \eqref{eq:claim-B-3}
by the same discussion as the proof of
Lemma \ref{prop:integral-for-Sm} above.
Then, f\/inally we obtain

\begin{lem} \label{prop:partial-taufunction-conjecture}
The equality \eqref{eq:claim-B-4} is true; namely, we have
\begin{gather}\label{eq:claim-B-5}
\frac{dF_{(k+2)/2}}{dt} = \sigma_{2k+2}.
\end{gather}
\end{lem}

\begin{proof}
It follows from
the equality \eqref{eq:compare-coefficients}
(for the odd number $m = k+1$) and Lemma~\ref{lem:Gg1} that
\begin{gather} \label{eq:S-rec-section3}
2 \frac{\p S_{0}}{\p x}\frac{\p S_{k+2}}{\p x} +
\sum_{\substack{a+b=k+2 \\ a, b \ge 1}}
\frac{\p S_{a}}{\p x}\frac{\p S_{b}}{\p x} +
\frac{\p^{2} S_{k+1}}{\p x^{2}}
- 2\frac{\p S_{k+1}}{\p t}
 = 2 \frac{d F_{(k+2)/2}}{dt}
\end{gather}
holds.
On the other hand, we know that $P_{k+2}$ satisf\/ies
\begin{gather} \label{eq:P-rec-section3}
2 P_{0}P_{k+2} +
\sum_{\substack{a+b=k+2 \\ a, b \ge 1}}
P_{a}P_{b} +
\frac{\p P_{k+1}}{\p x} - \left[ \frac{\hbar}{x-q}
\left( \hbar P-p \right) \right]_{\hbar^{k+2}}
 = 2 \sigma_{k+2}
\end{gather}
(cf.~\eqref{eq:recursion-for-Pm}).
Under our assumption, comparing \eqref{eq:S-rec-section3}
and \eqref{eq:P-rec-section3}, we have
\begin{gather} \label{eq:final-remark}
\frac{\p S_{0}}{\p x} \left( \frac{\p S_{k+2}}{\p x} - P_{k+2} \right)
 = \frac{d F_{(k+2)/2}}{dt} - \sigma_{k+2}.
\end{gather}
Note that the right hand-side doesn't depend on~$x$.
Then, thanks to the fact
\begin{gather*}
\frac{\p S_{0}}{\p x} \biggr|_{x=q_{0}}
= 0
\end{gather*}
and the holomorphicity of $S_{m}(x)$ and $P_{m}(x)$
at the double turning point $x = q_0$
(see Theorem~\ref{thm:Podd-and-t-derivative}),
we have the desired equality~\eqref{eq:claim-B-5}
by substituting $x = q_{0}$ into~\eqref{eq:final-remark}.
\end{proof}

This completes the proof of $(B)$ and Theorem~\ref{thm:induction-scheme}.
Thus we have proved Theorems~\ref{conj:2} and~\ref{conj:1}.

\begin{rem}
Since the spectral curve~\eqref{eq:sp-curve} has only
one branch point, we have
\begin{gather} \label{eq:v-and-infty}
\int^{x}_{v}P_{m}(x',t)dx' = \int^{x}_{\infty}P_{m}(x',t)dx'
\end{gather}
for all even $m \ge 2$.
This implies that the WKB solution~\eqref{eq:wave-function}
def\/ined by the topological recursion coincides with
the WKB solution~\eqref{eq:WKB-IM} constructed
in Section~\ref{subsection:WKB-for-scalar-Lax}.
However, the above equality~\eqref{eq:v-and-infty}
may not hold for
other Painlev\'e equations since their spectral curves have
more branch points in general.
\end{rem}

\appendix

\section[Alternative def\/inition of the $\tau$-function by Jimbo--Miwa--Ueno]{Alternative def\/inition of the $\boldsymbol{\tau}$-function\\ by Jimbo--Miwa--Ueno}
\label{section:relation-of-two-conjectures}

There is another def\/inition of $\tau$-function
\eqref{eq:tau} in terms of the formal
solution \eqref{eq:WKB-IM-mat} of the isomonodromy system.
\begin{prop}[\protect{\cite[Section~5]{JMU}; see also \cite[Section~4.2]{BBE12}
and \cite[Section~1.5]{Borot-Eynard}}]
The $\tau$-function satisfies
\begin{gather} \label{eq:alt-def-tau}
\frac{d}{dt} \log \tau(t,\hbar) = -2 \Res_{x=\infty} \left(
\frac{1}{\hbar} \frac{\partial T_{\infty}}{\partial t}(x,t)
{\mathcal W}_{1}(x,t,\hbar) dx \right),
\end{gather}
where
\begin{gather*}
T_{\infty}(x,t) := \frac{4x^{5/2}}{5} + t x^{1/2}
\end{gather*}
$($which is the divergent part of $\int^{x}P_{0}^{(+)}(x',t)dx'$
as $x\rightarrow \infty)$, and
\begin{gather*}
{\mathcal W}_{1}(x,t,\hbar) =
\frac{\p \psi_{+}}{\partial x}(x,t,\hbar) \tilde{\psi}_{-}(x,t,\hbar)
- \frac{\p \tilde{\psi}_{+} }{\p x}(x,t,\hbar){\psi}_{-}(x,t,\hbar).
\end{gather*}
\end{prop}

\begin{proof}
It follows from the def\/inition~\eqref{eq:WKB-IM-mat} of $\Psi$ that
\begin{gather*} 
{\mathcal W}_{1}(x,t,\hbar) =
P^{(+)}(x,t,\hbar) +
\frac{A_{12}(x,t,\hbar)}{2\hbar P_{\rm odd}(x,t,\hbar)}
\frac{\p}{\p x}\left( \frac{\hbar P^{(+)} (x,t,\hbar)
- A_{11}(x,t,\hbar)}
{A_{12}(x,t,\hbar)} \right).
\end{gather*}
Then, the asymptotics \eqref{eq:asymp-P}
of $P^{(\pm)}(x,t,\hbar)$ implies that
\begin{gather*}
{\mathcal W}_{1}(x,t,\hbar)
= \frac{2}{\hbar}x^{3/2} + \frac{t}{2\hbar}x^{-1/2} +
\frac{\sigma(t,\hbar)}{2\hbar} x^{-3/2}
+ O\big(x^{-2}\big)
\end{gather*}
holds when $x \rightarrow \infty$, and
thus we have \eqref{eq:alt-def-tau}.
\end{proof}

\subsection*{Acknowledgements}
The authors are grateful to
Motohico Mulase for many valuable comments,
discussion and continuous encouragements.
They also thank
Olivia Dumitrescu and Bertrand Eynard
for helpful comments.
K.I. work is supported by the JSPS for
Advancing Strategic International Networks to Accelerate
the Circulation of Talented Researchers
``Mathematical Science of Symmetry, Topology and Moduli,
Evolution of International Research Network
based on Osaka City University Advanced Mathematical Institute
(OCAMI)''.
A.S.~work is supported by UC Davis under the Graduate
Research Mentorship fellowship.
This article is written during the
K.I.~stay at The University of California, Davis.
K.I.~would also like to thank the
institute for its support and hospitality.

\pdfbookmark[1]{References}{ref}
\LastPageEnding


\begin{thebibliography}{99}
\footnotesize\itemsep=0pt

\bibitem{Agana}
Aganagic M., Cheng M.C.N., Dijkgraaf R., Kref\/l D., Vafa C., Quantum geometry of
  ref\/ined topological strings, \href{http://dx.doi.org/10.1007/JHEP11(2012)019}{\textit{J.~High Energy Phys.}} \textbf{2012}
  (2012), no.~11, 019, 53~pages, \href{http://arxiv.org/abs/1105.0630}{arXiv:1105.0630}.

\bibitem{Agana2}
Aganagic M., Dijkgraaf R., Klemm A., Mari{\~n}o M., Vafa C., Topological
  strings and integrable hierarchies, \href{http://dx.doi.org/10.1007/s00220-005-1448-9}{\textit{Comm. Math. Phys.}} \textbf{261}
  (2006), 451--516, \href{http://arxiv.org/abs/hep-th/0312085}{hep-th/0312085}.

\bibitem{Umeta}
Aoki T., Honda N., Umeta Y., On a construction of general formal solutions for
  equations of the f\/irst {P}ainlev\'e hierarchy~{I}, \href{http://dx.doi.org/10.1016/j.aim.2012.12.011}{\textit{Adv. Math.}}
  \textbf{235} (2013), 496--524.

\bibitem{AKT-P}
Aoki T., Kawai T., Takei Y., W{KB} analysis of {P}ainlev\'e transcendents with
  a large parameter. {II}.~{M}ultiple-scale analysis of {P}ainlev\'e
  transcendents, in Structure of Solutions of Dif\/ferential Equations
  ({K}atata/{K}yoto, 1995), World Sci. Publ., River Edge, NJ, 1996, 1--49.

\bibitem{BBE12}
Berg{\'e}re M., Borot G., Eynard B., Rational dif\/ferential dystems, loop
  equations, and application to the $q$th reductions of {KP}, \href{http://dx.doi.org/10.1007/s00023-014-0391-8}{\textit{Ann.
  Henri Poincar\'e}} \textbf{16} (2015), 2713--2782, \href{http://arxiv.org/abs/1312.4237}{arXiv:1312.4237}.

\bibitem{BE}
Berg{\'e}re M., Eynard B., Determinantal formulae and loop equations,
  \href{http://arxiv.org/abs/0901.3273}{arXiv:0901.3273}.

\bibitem{Borot-Eynard}
Borot G., Eynard B., {T}racy--{W}idom GUE law and symplectic invariants,
  \href{http://arxiv.org/abs/1011.1418}{arXiv:1011.1418}.

\bibitem{Costin08}
Costin O., Asymptotics and {B}orel summability, \textit{Chapman {\rm \&} Hall/CRC
  Monographs and Surveys in Pure and Applied Mathematics}, Vol.~141, CRC Press,
  Boca Raton, FL, 2009.

\bibitem{DGZJ}
Di~Francesco P., Ginsparg P., Zinn-Justin J., {$2$}{D} gravity and random
  matrices, \href{http://dx.doi.org/10.1016/0370-1573(94)00084-G}{\textit{Phys. Rep.}} \textbf{254} (1995), 1--133,
  \href{http://arxiv.org/abs/hep-th/9306153}{hep-th/9306153}.

\bibitem{DFM}
Dijkgraaf R., Fuji H., Manabe M., The volume conjecture, perturbative knot
  invariants, and recursion relations for topological strings, \href{http://dx.doi.org/10.1016/j.nuclphysb.2011.03.014}{\textit{Nuclear
  Phys.~B}} \textbf{849} (2011), 166--211, \href{http://arxiv.org/abs/1010.4542}{arXiv:1010.4542}.

\bibitem{DM14}
Dumitrescu O., Mulase M., Quantum curves for {H}itchin f\/ibrations and the
  {E}ynard--{O}rantin theory, \href{http://dx.doi.org/10.1007/s11005-014-0679-0}{\textit{Lett. Math. Phys.}} \textbf{104} (2014),
  635--671, \href{http://arxiv.org/abs/1310.6022}{arXiv:1310.6022}.

\bibitem{DM14-2}
Dumitrescu O., Mulase M., Quantization of spectral curves for meromorphic
  {H}iggs bundles through topological recursion, \href{http://arxiv.org/abs/1411.1023}{arXiv:1411.1023}.

\bibitem{DBMNPS}
Dunin-Barkowski P., Mulase M., Norbury P., Popolitov A., Shadrin S., Quantum
  spectral curve for the {G}romov--{W}itten theory of the complex projective
  line, \href{http://dx.doi.org/10.1515/crelle-2014-0097}{\textit{J.~Reine Angew. Math.}}, {t}o appear, \href{http://arxiv.org/abs/1312.5336}{arXiv:1312.5336}.

\bibitem{Eynard14}
Eynard B., Topological recursion and quantum curves, {T}alk given in the
  workshop ``Quantum curves, Hitchin systems, and the Eynard--Orantin theory'',
  American Institute of Mathematics, Palo Alto, September 2014.

\bibitem{Eynard-book}
Eynard B., Counting surfaces,
\href{http://dx.doi.org/10.1007/978-3-7643-8797-6}{\textit{Progress in Mathematical Physics}},
Vol.~70, Birkh\"auser, Basel, 2016.

\bibitem{EO07}
Eynard B., Orantin N., Invariants of algebraic curves and topological
  expansion, \href{http://dx.doi.org/10.4310/CNTP.2007.v1.n2.a4}{\textit{Commun. Number Theory Phys.}} \textbf{1} (2007), 347--452,
  \href{http://arxiv.org/abs/math-ph/0702045}{math-ph/0702045}.

\bibitem{FIKN}
Fokas A.S., Its A.R., Kapaev A.A., Novokshenov V.Yu., Painlev\'e transcendents.
  The {R}iemann--{H}ilbert approach, \href{http://dx.doi.org/10.1090/surv/128}{\textit{Mathematical Surveys and
  Monographs}}, Vol.~128, Amer. Math. Soc., Providence, RI, 2006.

\bibitem{GS}
Gukov S., Su{\l}kowski P., A-polynomial, {B}-model, and quantization,
  \href{http://dx.doi.org/10.1007/JHEP02(2012)070}{\textit{J.~High Energy Phys.}} \textbf{2012} (2012), no.~2, 070, 57~pages,
  \href{http://arxiv.org/abs/1108.0002}{arXiv:1108.0002}.

\bibitem{IM}
Iwaki K., Marchal O., Painlev\'e 2 equation with arbitrary monodromy parameter,
  topological recursion and determinantal formulas, \href{http://arxiv.org/abs/1411.0875}{arXiv:1411.0875}.

\bibitem{JMU}
Jimbo M., Miwa T., Ueno K., Monodromy preserving deformation of linear ordinary
  dif\/ferential equations with rational coef\/f\/icients. {I}.~{G}eneral theory and
  {$\tau $}-function, \href{http://dx.doi.org/10.1016/0167-2789(81)90013-0}{\textit{Phys.~D}} \textbf{2} (1981), 306--352.

\bibitem{JM2}
Jimbo M., Miwa T., Monodromy preserving deformation of linear ordinary
  dif\/ferential equations with rational coef\/f\/icients.~{II}, \href{http://dx.doi.org/10.1016/0167-2789(81)90021-X}{\textit{Phys.~D}}
  \textbf{2} (1981), 407--448.

\bibitem{Joshi-Kitaev}
Joshi N., Kitaev A.V., On {B}outroux's tritronqu\'ee solutions of the f\/irst
  {P}ainlev\'e equation, \href{http://dx.doi.org/10.1111/1467-9590.00187}{\textit{Stud. Appl. Math.}} \textbf{107} (2001),
  253--291.

\bibitem{Kamimoto-Koike}
Kamimoto S., Koike T., On the {B}orel summability of 0-parameter solutions of
  nonlinear ordinary dif\/ferential equations, in Recent Development of
  Micro-Local Analysis for the Theory of Asymptotic Analysis, \textit{RIMS
  K\^oky\^uroku Bessatsu}, Vol.~B40, Res. Inst. Math. Sci. (RIMS), Kyoto, 2013,
  191--212.

\bibitem{Kap}
Kapaev A.A., Asymptotic behavior of the solutions of the {P}ainlev\'e equation
  of the f\/irst kind, \textit{Differential Equations} \textbf{24} (1988),
  1107--1115.

\bibitem{KT-PI}
Kawai T., Takei Y., W{KB} analysis of {P}ainlev\'e transcendents with a large
  parameter.~{I}, \href{http://dx.doi.org/10.1006/aima.1996.0016}{\textit{Adv. Math.}} \textbf{118} (1996), 1--33.

\bibitem{KT05}
Kawai T., Takei Y., Algebraic analysis of singular perturbation theory,
  \textit{Translations of Mathematical Monographs}, Vol.~227, Amer. Math. Soc.,
  Providence, RI, 2005.

\bibitem{KNS}
Kawakami H., Nakamura A., Sakai H., Degeneration scheme of 4-dimensional
  Painlev\'e-type equations, \href{http://arxiv.org/abs/1209.3836}{arXiv:1209.3836}.

\bibitem{Kont}
Kontsevich M., Intersection theory on the moduli space of curves and the matrix
  {A}iry function, \href{http://dx.doi.org/10.1007/BF02099526}{\textit{Comm. Math. Phys.}} \textbf{147} (1992), 1--23.

\bibitem{MS12}
Mulase M., Su{\l}kowski P., Spectral curves and the Schr\"odinger equations for
  the {E}ynard--{O}rantin recursion, \href{http://arxiv.org/abs/1210.3006}{arXiv:1210.3006}.

\bibitem{Nakamura}
Nakamura A., Autonomous limit of 4-dimensional Painlev\'e-type equations and
  degeneration of curves of genus two, \href{http://arxiv.org/abs/1505.00885}{arXiv:1505.00885}.

\bibitem{Nor}
Norbury P., Quantum curves and topological recursion, \href{http://arxiv.org/abs/1502.04394}{arXiv:1502.04394}.

\bibitem{Okamoto}
Okamoto K., Polynomial {H}amiltonians associated with {P}ainlev\'e
  equations.~{I}, \href{http://dx.doi.org/10.3792/pjaa.56.264}{\textit{Proc. Japan Acad. Ser.~A Math. Sci.}} \textbf{56}
  (1980), 264--268.

\bibitem{Olsha}
Olshanetsky M.A., Painlev\'e type equations and {H}itchin systems, in
  Integrability: the {S}eiberg--{W}itten and {W}hitham Equations ({E}dinburgh,
  1998), Gordon and Breach, Amsterdam, 2000, 153--174,
  \href{http://arxiv.org/abs/math-ph/9901019}{math-ph/9901019}.

\bibitem{Painleve}
Painlev{\'e} P., Sur les \'equations dif\/f\'erentielles du second ordre et
  d'ordre sup\'erieur dont l'int\'egrale g\'en\'erale est uniforme,
  \href{http://dx.doi.org/10.1007/BF02419020}{\textit{Acta Math.}} \textbf{25} (1902), 1--85.

\bibitem{Takasaki}
Takasaki K., Spectral curves and {W}hitham equations in isomonodromic problems
  of {S}chlesinger type, \textit{Asian~J. Math.} \textbf{2} (1998), 1049--1078,
  \href{http://arxiv.org/abs/solv-int/9704004}{solv-int/9704004}.

\bibitem{Takei}
Takei Y., An explicit description of the connection formula for the f\/irst
  {P}ainlev\'e equation, in Toward the Exact {WKB} Analysis of Dif\/ferential
  Equations, Linear or Non-Linear ({K}yoto, 1998), Kyoto University Press,
  Kyoto, 2000, 271--296.

\end{thebibliography}
\end{document}